\documentclass[lettersize,onecolumn]{IEEEtran}

\usepackage{amsmath,amsfonts}
\usepackage{algorithmic}
\usepackage{algorithm}
\usepackage{array}
\usepackage{subcaption}
\usepackage{textcomp}
\usepackage{stfloats}
\usepackage{url}
\usepackage{verbatim}
\usepackage{graphicx}
\usepackage[nocompress]{cite}
\usepackage{newtxtext}
\usepackage{amssymb, amsmath,threeparttable, bbding, url,xcolor}
\usepackage{booktabs}
\usepackage{circuitikz}
\usepackage{multirow}
\newtheorem{theorem}{Theorem}[section]
\newtheorem{lemma}[theorem]{Lemma}
\newtheorem{proposition}[theorem]{Proposition}
\newtheorem{example}[theorem]{Example}
\newtheorem{remark}[theorem]{Remark}
\newtheorem{corollary}[theorem]{Corollary}
\newtheorem{definition}{Definition}

\newcommand{\blue}{\color{blue}}

\newcommand{\supp}{\operatorname{supp}(\ell)}

\ctikzset{logic ports=ieee}

\newcommand{\F}{{\mathbb{F}}}

\begin{document}
	\title{Shift-Invariant Permutations with Polynomial Composition: Construction and Algebraic Structure}
	%
	%
	%
\author{ Cheng~Lyu,\,\,
		Dabin~Zheng,\,\,
		Mu~Yuan,\,\,
		Lei~Hu\,\,
\thanks{ C. Lyu, D. Zheng and M. Yuan are with Hubei Key Laboratory of Applied Mathematics, Faculty of Mathematics and Statistics, Hubei University. Dabin Zheng is also with Key Laboratory of Intelligent Sensing System and Security (Hubei University), Ministry of Education, Wuhan, 430062, China. E-mail: chenglyu@139.com; dzheng@hubu.edu.cn; yuanmu847566@outlook.com. Lei Hu is with Key Laboratory of Cyberspace Security Defense, Institute of Information Engineering, Chinese Academy of Sciences, Beijing, 100093, China. Email: hulei@iie.ac.cn. The corresponding author is Dabin Zheng.}}

\markboth{IEEE Transactions on Information Theory,~Vol.~, No.~}%
{{\blue Lyu} \MakeLowercase{\textit{et al.}}: Shift-Invariant Permutations with Polynomial Composition: Construction and Algebraic Structure}
%



\maketitle

\begin{abstract}
Shift-invariant permutations of \(\mathbb F_2^n\) are particularly attractive in symmetric cryptography because of their regularity and implementation efficiency. Constructing such permutations with an explicit algebraic structure remains a challenging problem. In this paper, we study a family of shift-invariant transformations of \(\mathbb F_2^n\) constructed from a recursively defined sequence of mappings \(\{\gamma_j\}_{j\ge 0}\) associated with landscapes. First, we prove that, under an explicit dimension condition relating \(n\) to the support of the landscape, the sequence \(\{\gamma_j\}_{j\ge 0}\) has the polynomial composition property if and only if the corresponding landscape is quasi-conserved on \(\mathbb F_2^n\) with a nonzero special index. Then, we determine the eventual zero or periodic behavior of \(\{\gamma_j\}_{j\ge0}\), including the exact index at which the sequence either first vanishes or becomes periodic, together with the least eventual period, and establish linear independence until the first linear relation occurs. When the polynomial composition property holds, we establish an explicit isomorphism between the group of permutation elements in the monoid generated by these mappings and the unit group of a quotient ring \(\mathbb F_2[z]/\langle p(z)\rangle\), where \(p(z)\) is a monomial or a binomial determined by the landscape. This framework reduces the tasks of testing bijectivity, computing iterates, and determining compositional inverses and orders of permutation elements to polynomial computations. Finally, we apply the results to mappings represented by binomials and trinomials in the quotient ring, obtaining explicit criteria and formulas, and tabulate more than one hundred representative constructions of shift-invariant permutations, recovering and unifying several previously studied families.
\end{abstract}
	
\begin{IEEEkeywords}
Permutations, S-boxes, shift-invariant transformations.
\end{IEEEkeywords}

\section{Introduction}

A vectorial Boolean function \(F\colon\mathbb F_2^n\to\mathbb F_2^n\) is called \emph{shift-invariant} if it commutes with the cyclic shift operator \(S\), i.e., \(F\circ S=S\circ F\), where \(S\) maps \((x_0,x_1,\ldots,x_{n-1})\) to \((x_1,\ldots,x_{n-1},x_0)\). Equivalently, the coordinate functions of \(F\) are cyclic shifts of a single Boolean function, i.e., $F(x) =  \bigl(f(x),f(Sx),\ldots,f(S^{n-1}x)\bigr)$ for some Boolean function \( f\colon\mathbb F_2^n\to\mathbb F_2\).
In the infinite setting \(\mathbb F_2^{\mathbb Z}\), the analogous notion is that of a translation-invariant local rule, i.e., a one-dimensional cellular automaton. When \(f\) depends only on a fixed local neighborhood, the resulting mappings are the periodic-boundary-case counterparts of one-dimensional cellular automata (see, e.g., \cite{mariot2024}). When such a mapping is bijective, it is called a shift-invariant permutation. Such transformations, especially bijective ones, are particularly attractive in symmetric cryptography, because their regular structure enables highly efficient bit-sliced software implementations and low-latency hardware realizations, while the entire mapping is completely determined by a single coordinate rule. However, constructing shift-invariant permutations and deriving their algebraic properties are generally difficult tasks. For arbitrary local rules, to the best of our knowledge, no simple closed-form criterion for invertibility is known. Moreover, uniform explicit formulas for inverses, orders, or cycle structures are generally difficult to obtain. It is therefore natural to restrict attention to well-behaved subclasses that nevertheless capture many practically relevant constructions.

A particularly fruitful subclass consists of local rules of multiplicative type, which includes the celebrated \(\chi\) mapping and its numerous generalizations. The classical \(\chi\) mapping from \(\mathbb F_2^n\) to itself, introduced by Daemen~\cite{daemen1995cipher}, is defined by
\[
\chi(x)_i=x_i+(x_{i+1}+1)x_{i+2},
\]
where all coordinate indices are taken modulo \(n\). Owing to its high efficiency and its well-studied cryptographic properties, \(\chi\) has been adopted in several prominent standards: it is used directly in \textsc{Keccak}-\(f\)~\cite{bertoni2008keccak} (the basis of SHA-3~\cite{nist2015sha}), while Ascon~\cite{dobraunig2021ascon} (the NIST lightweight cryptography standard~\cite{nist2023ascon}) uses an S-box that is affine-equivalent to $\chi$. The mapping \(\chi\) has been extensively investigated from both cryptanalytic and algebraic viewpoints; see, e.g.,~\cite{belkheyar2025chi,biryukov2014asasa,daemen2021comput,feng2026generalization,graner2024bijectivity,liu2022inverse,mella2023diff,schoone2024algebraic,schoone2024state}.

A major breakthrough in understanding the algebraic structure behind $\chi$ was made by Kriepke and Kyureghyan~\cite{kriepke2024algebraic}. They established an explicit isomorphism between the monoid generated by \(\chi\) together with certain related functions, on the one hand, and the unit group of the quotient ring \(\mathbb F_2[z]/\langle z^{(n+1)/2}\rangle\) for odd \(n\), on the other hand. Under this correspondence, the shift operator \(S\) and the Hadamard product translate the composition of functions into polynomial multiplication, thereby reducing the permutation property, order, iterates, and inverse of \(\chi\) to simple computations in a quotient polynomial ring.

It is known that \(\chi\) permutes \(\mathbb F_2^n\) only for odd \(n\). To construct shift-invariant permutations over \(\mathbb F_2^n\) for even \(n\), building on the powerful framework proposed in~\cite{kriepke2024algebraic}, we recently generalized \(\chi\) to a broader form \(\chi_{n,m}\), and obtained a larger class of permutations of \(\mathbb F_2^n\) that covers both odd and even dimensions \(n\)~\cite{lyu2025generalized}. Meanwhile, for even \(n\), Kriepke and Kyureghyan~\cite{kriepke2025siblings} used the functions \(\gamma_{2k}\), which are built from the cyclic shift operator and the Hadamard product, to span a vector space over \(\mathbb F_2\), and this vector space contains an abelian group of permutations isomorphic to the unit group of \(\mathbb F_2[z]/\langle z^n+z^{n/2}\rangle\). Zhao et al.~\cite{zhao2026family} use bifix-free sequences to study the set of defining functions of shift-invariant mappings, and establish an isomorphism between this set and a polynomial quotient ring, and thereby construct shift-invariant permutations and obtain their arithmetic properties. Other recent work studies reversible cellular automata and constructs explicit inverses for particular shift-invariant mappings~\cite{haugland2026new,haugland2025shift,liu2026finding}.

The key to the success of the generalizations in~\cite{kriepke2025siblings,lyu2025generalized} is that the affine families formed by the identity plus linear combinations of functions in \(\{\gamma_{2k}\}_{k\ge1}\), \(\{\theta_{m,k}\}_{k\ge1}\), respectively, are closed under composition. In other words, the composition of any two functions from such a family can still be expressed in the same form. Based on this important observation, Kriepke in~\cite{kriepke2026shift} studied which sequences of functions possess such a property by using landscapes, a notion that goes back to Daemen~\cite{daemen1995cipher}. In his work, a landscape \(\ell\) is defined as a finite string over \(\{0,1,-,*\}\), which is associated with a shift-invariant function from $\mathbb F_2^{\mathbb Z}$ to itself. From this landscape he recursively constructs a general sequence \(\{\gamma_j^{(\ell,q)}\}_{j\geq 0}\) of functions from $\mathbb F_2^{\mathbb Z}$ to itself. He stated that for \(q\neq 0\), the constructed function sequence \(\{\gamma_j^{(\ell,q)}\}_{j\geq 0}\) satisfies the polynomial composition property, i.e., for every $j,s\ge0$,
\begin{equation}\label{eq:pcp}
\gamma_j^{(\ell,q)}\circ\left(\operatorname{id}+\sum_{i=1}^{s}\alpha_i\gamma_i^{(\ell,q)}\right)
=\sum_{i=0}^{s}\alpha_i\gamma_{i+j}^{(\ell,q)},
\qquad \alpha_0=1,\quad \alpha_1,\ldots,\alpha_s\in\mathbb F_2,
\end{equation}
if and only if \(\ell\) is quasi-conserved with special index \(q\), but he did not give a proof of this claim in~\cite{kriepke2026shift}. This connection reveals the essential characteristics of shift-invariant mappings that behave like polynomials, and further demonstrates that the algebraic framework proposed in~\cite{kriepke2024algebraic} is highly effective for studying shift-invariant permutations on $\mathbb F_2^n$. The case $q=0$ is allowed in~\cite{kriepke2026shift}, but in this case, quasi-conservation may not imply the polynomial composition property of \(\{\gamma_j\}_{j\geq 0}\), and a counterexample is given in the Appendix.

However, Kriepke in~\cite{kriepke2026shift} treated functions on $\mathbb F_2^n$ by restriction to the $n$-periodic configurations in $\mathbb F_2^{\mathbb Z}$, and did not directly discuss the equivalence between the landscape function 
$\ell$ on $\mathbb F_2^n$  having the quasi-conserved property and the corresponding function sequence $\{\gamma_j\}_{j\geq 0}$ having the polynomial composition property, nor did he discuss the connection between quasi-conservation of 
$\ell$ and the polynomial composition property of the corresponding function sequence $\{\gamma_j\}_{j\geq 0}$ on these two spaces. In this paper, we directly establish the equivalence between quasi-conservation of $\ell$ and the polynomial composition property of $\{\gamma_j\}_{j\geq 0}$ on $\mathbb F_2^n$ under an explicit dimension condition, and also discuss the equivalence between quasi-conservation of $\ell$ and the polynomial composition property of $\{\gamma_j\}_{j\geq 0}$ on the two spaces $\mathbb F_2^n$ and $\mathbb F_2^{\mathbb Z}$ in the Appendix. To this end,  we first represent a landscape in the following equivalent form:
\begin{equation}\label{eq:landscape-function}
\ell=\bigodot_{r\in R}(\mathbf1+S^r)\bigodot_{t\in T}S^t,
\end{equation}
where \(\bigodot\) denotes the Hadamard product, \(R,T\subseteq\{-k,-k+1,\ldots,m-1,m\}\setminus\{0\}\) are disjoint finite nonempty sets, and $k,m$ are the smallest nonnegative integers such that the inclusion \(R,T\subseteq\{-k,\ldots,m\}\setminus\{0\}\) holds. Using the complementation symmetry proved below, we may take a special index $q\in T$ and put
\begin{equation}\label{eq:ellq}
\ell^{(q)}=\bigodot_{r\in R}(\mathbf1+S^r)\bigodot_{t\in T\setminus\{q\}}S^t.
\end{equation}
Following~\cite[Definition~3.6]{kriepke2026shift}, we study the following sequence on $\mathbb F_2^n$ with $n>m+k$:
\begin{equation}\label{eq:gammaj}
\gamma_0=\operatorname{id},\qquad \gamma_1=\ell,\qquad
\gamma_{j+1}=\ell^{(q)}\odot (S^q\circ \gamma_j), \quad(j\ge1).
\end{equation}
Our main contributions are as follows.
\begin{enumerate}
\item[(1)] We establish a necessary and sufficient condition for the landscape \(\ell\) determined by \((R,T)\) to be quasi-conserved with a given special index \(q\), and prove that the sequence of functions on $\mathbb F_2^n$ in~\eqref{eq:gammaj} has the polynomial composition property in~\eqref{eq:pcp} if and only if the landscape $\ell$ is quasi-conserved on $\mathbb F_2^n$, provided that $n>\max\{2m+k,2k+m\}$. This is a refinement of Theorem~3.9 in~\cite{kriepke2026shift}.
\item[(2)] For all landscapes in this class that are quasi-conserved on $\mathbb F_2^n$, we determine the long-term behavior of the function sequence \( \{\gamma_j\}_{j\geq 0}\). We show that, depending on the arithmetic relations between \(R,T,q\) and \(n\), the sequence is either eventually zero or remains nonzero and becomes periodic. We determine the exact first zero or periodic starting index and the least eventual period, and prove linear independence up to the first relation.
\item[(3)] When the polynomial composition property holds, we use this algebraic framework to derive permutation criteria and formulas for iterates in the corresponding monoid, and for compositional inverses and orders of its permutation elements. For the mappings $\operatorname{id}+\gamma_1$ and $\operatorname{id}+\gamma_1+\gamma_2$ (see Section~\ref{sec:families}), we also determine the dimensions \(n\) on which they are permutations. These results provide a unified treatment of several previously studied families of shift-invariant permutations. As a direct application, we obtain over one hundred representative constructions of shift-invariant permutations over \(\mathbb F_2^n\). These are systematically tabulated in Tables~\ref{tab:landscapes}, \ref{tab:landscapes-Tq} and~\ref{tab:landscapes-Tqt} at the end of Section~\ref{sec:families}.
\end{enumerate}
The present work thus not only resolves the problem of characterizing quasi-conserved landscapes for a wide family, but also demonstrates the effectiveness and versatility of the algebraic framework initiated in~\cite{kriepke2024algebraic} and extended in~\cite{kriepke2026shift}. Our results encompass and explain several previously studied families of shift-invariant permutations, including those from~\cite{lyu2025generalized,kriepke2025siblings,liu2026finding,haugland2026new}. Moreover, they yield a large number of further such permutations and allow us to determine their inverses and algebraic structures explicitly.

The paper is organized as follows. Section~\ref{sec:prelim} recalls the necessary preliminaries on shift-invariant functions, the Hadamard product, the quasi-conserved landscapes and the polynomial composition property of the sequence of functions on
$\mathbb F_2^n$. Section~\ref{sec3} contains our main results: we first characterize the polynomial composition property of the sequence $\{\gamma_j\}_{j\geq 0}$ in terms of quasi-conservation of $\ell$, then determine the long-term behavior and linear independence of the sequence $\{\gamma_j\}_{j\geq 0}$, and construct shift-invariant permutations on $\mathbb F_2^n$ and study the algebraic structure of these permutations by establishing a homomorphic mapping between the vector space generated by the sequence $\{ \gamma_j\}_{j\geq 0}$ and a quotient ring. In Section~\ref{sec:families} we construct some explicit shift-invariant permutations of $\mathbb F_2^n$. Section~\ref{sec:con} gives the conclusion and future work. The Appendix proves the equivalence of quasi-conservation and polynomial composition on $\mathbb F_2^{\mathbb Z}$, relates it to the finite-dimensional result, and gives a counterexample for $q=0$.


\section{Preliminaries}\label{sec:prelim}
Let $[i,j]$ denote the set of integers $\{i,i+1,\ldots,j\}$ for integers $i\le j$. We define the cyclic shift operator $S\colon\mathbb F_2^n\to\mathbb F_2^n$ by $S(x_0,x_1,\ldots,x_{n-1})=(x_1,\ldots,x_{n-1},x_0)$. A mapping $F\colon\mathbb F_2^n\to\mathbb F_2^n$ is called \emph{shift-invariant} if $S(F(x))=F(S(x))$ for all $x$. All coordinate subscripts are read modulo $n$ unless stated otherwise.

Let $\mathbf0$ and $\mathbf1$ be the vectors with all components equal to $0$ and $1$, respectively. The same symbols denote the corresponding constant mappings when the meaning is clear. For $x,y\in\mathbb F_2^n$ we define their component-wise product $z=x\odot y$ by $z_i=x_i y_i$ for $i\in[0,n-1]$. The operation $\odot$ is associative, commutative and distributive over addition and has $\mathbf1$ as its identity. Observe that $S$ is linear and $S(x\odot y)=S(x)\odot S(y)$. We extend $\odot$ to mappings by $(F\odot G)(x)=F(x)\odot G(x)$. With this notation the classical $\chi$ mapping can be written as $\chi=\operatorname{id}+(\mathbf1+S)\odot S^2$. For brevity, we write $S^iF$ for the composition $S^i\circ F$.

Fix disjoint finite sets $R,T\subseteq\mathbb Z$, not both empty, and let $k,m$ be the smallest nonnegative integers such that $R,T\subseteq[-k,m]$. We assume $n>m+k$. A function of the form~\eqref{eq:landscape-function} is called a \emph{landscape function}. Equivalently,
\begin{equation}\label{eq:landscape-coordinate}
\ell(x)_i=\prod_{r\in R}(x_{i+r}+1)\prod_{t\in T}x_{i+t}.
\end{equation}
The inequality $n>m+k$ ensures that the offsets in $[-k,m]$ are distinct modulo $n$.

Daemen~\cite{daemen1995cipher} introduced landscape strings over $\{0,1,-,*\}$ to describe local transformations. We also use a string $\hat\ell$ over $\{0,1,-,*\}$: $0$ indicates offsets in $R$, $1$ offsets in $T$, and $-$ the remaining noncentral offsets in $[-k,m]$. At the center we write $0$, $1$, or $*$ according as $0\in R$, $0\in T$, or $0\notin R\cup T$. Conversely, every such string containing at least one $0$ or $1$ determines a landscape function via the above Hadamard product. For example, the classical transformation $\chi$ has coordinate function \(\chi(x)_i=x_i+(1+x_{i+1})x_{i+2}\). Its associated landscape function is $\ell=(\mathbf1+S)\odot S^2$, with $R=\{1\}$, $T=\{2\}$, and $\hat\ell={*}01$. As another example, the landscape function
$\ell=(S^{-1}+\mathbf1)\odot(S^2+\mathbf1)\odot S^{-2}\odot S$ has $R=\{-1,2\}$, $T=\{-2,1\}$, and $\hat\ell=10{*}10$.

The following definition is a finite-dimensional version of Kriepke's definition in~\cite{kriepke2026shift}.
\begin{definition}\label{def:quasi-conserved}
The set $R\cup T$ is called the \emph{support} of the landscape function $\ell$, denoted by $\operatorname{supp}(\ell)$. A landscape function $\ell$ is called \emph{conserved} if $S^\delta\ell\odot\ell=\mathbf0$ holds for every $\delta\in\operatorname{supp}(\ell)$, and \emph{quasi-conserved} with special index $q\in\operatorname{supp}(\ell)$ if the equality holds for all $\delta\in\operatorname{supp}(\ell)\setminus\{q\}$. For $q\in\supp$, let $\ell^{(q)}$ denote the local rule obtained from~\eqref{eq:landscape-function} by deleting the factor indexed by $q$.
\end{definition}

The relation to the corresponding definition on $\mathbb F_2^{\mathbb Z}$ in~\cite{kriepke2026shift} is discussed in the Appendix. The following lemma gives a modular criterion for quasi-conservation.

\begin{lemma}\label{lem:quasi-cons-set}
Assume $n>m+k$. For $\delta\in\supp\setminus\{q\}$, $S^\delta\ell\odot\ell=\mathbf0$ if and only if
\begin{equation}\label{eq:finite-qc-set}
(R+\delta)\cap T\ne\varnothing
\qquad\text{or}\qquad
(T+\delta)\cap R\ne\varnothing,
\end{equation}
where the above addition is taken modulo $n$. Equivalently, for some $r\in R$ and $t\in T$,
\[r+\delta\equiv t\pmod n \qquad\text{or}\qquad t+\delta\equiv r\pmod n.\]
Hence $\ell$ is quasi-conserved with special index $q$ if and only if
\eqref{eq:finite-qc-set} holds for every $\delta\in\supp\setminus\{q\}$.
\end{lemma}

\begin{IEEEproof}
For fixed $\delta$, $S^\delta\ell\odot\ell=\mathbf0$ exactly when its
factors contain both $S^w$ and $S^w+\mathbf1$ at the same coordinate. This is equivalent to one of
the two displayed congruences. If no such pair occurs, choose an input
satisfying all required zero and one coordinates; then $S^\delta\ell\odot\ell\ne\mathbf0$.
\end{IEEEproof}

Recall the definition introduced in \cite[Definition~3.8]{kriepke2026shift} as follows:
\begin{definition}\label{def:polynomial-composition-property}
The sequence $\{\gamma_j\}_{j\ge0}$ in~\eqref{eq:gammaj} satisfies the \emph{polynomial composition property} if, for every $j,s\ge0$ and all
$\alpha_0,\ldots,\alpha_s\in\mathbb F_2$ with $\alpha_0=1$,
\begin{equation}\label{eq:polynomial-composition-property}
\gamma_j\circ\left(\gamma_0+\sum_{i=1}^{s}\alpha_i\gamma_i\right)=\sum_{i=0}^{s}\alpha_i\gamma_{i+j}.
\end{equation}
\end{definition}

Throughout the paper, we assume that $R$ and $T$ are nonempty, $q\ne 0$, and $0\notin\supp$. The assumption $q\ne0$ is necessary in general (see the Appendix).
Under quasi-conservation, $q\ne0$ implies $0\notin\supp$. We can further assume that $q\in T$. Let $\bar\ell$ be obtained from $\ell$ by exchanging $R$ and $T$, let $J(x)=x+\mathbf1$, and let
$\{\bar\gamma_j\}$ be the corresponding sequence. Then $\bar\ell=\ell\circ J$ and $\bar\gamma_j=\gamma_j\circ J$ for $j\ge1$. If $f=\gamma_0+\sum_{i=1}^s\alpha_i\gamma_i$ and $\bar f=\gamma_0+\sum_{i=1}^s\alpha_i\bar\gamma_i$, then
$\bar f=J\circ f\circ J$ and $\bar\gamma_j\circ\bar f=(\gamma_j\circ f)\circ J$ for $j\ge1$. Thus complementation preserves both quasi-conservation and the polynomial composition property, and exchanging $R$ and $T$ moves $q$ from $R$ to $T$.

\section{Shift-Invariant Permutations from Quasi-Conserved Landscapes}\label{sec3}

In this section, we first relate quasi-conservation of a landscape $\ell$ to the polynomial composition property of $\{\gamma_j\}_{j\geq 0}$, and prove  an $\mathbb F_2^n$-version of Theorem 3.9 in~\cite{kriepke2026shift}. We then determine the long-term behavior and linear independence of $\{\gamma_j\}_{j\geq 0}$ and apply the quotient-ring composition method to describe shift-invariant permutations and their algebraic structure. We begin with the following closed form.

\begin{lemma}\label{lem:gamma-explicit}
The terms of the sequence defined in~\eqref{eq:gammaj} can also be represented as
\begin{equation}\label{eq:gamma-explicit}
\gamma_j=S^{jq}\bigodot_{i=0}^{j-1}S^{iq}\ell^{(q)}=S^{jq}\bigodot_{i=0}^{j-1}\left(\bigodot_{r\in R}(S^{iq+r}+\mathbf1)\bigodot_{t\in T\setminus\{q\}}S^{iq+t}\right),
\end{equation}
where the empty product for $j=0$ is $\mathbf1$.
\end{lemma}

\begin{IEEEproof}
The cases $j=0,1$ follow from the definitions. Suppose that the formula holds for some
$j\ge1$. Applying $S^q$ to the induction hypothesis and using
$S^q(F\odot G)=S^qF\odot S^qG$, we obtain
\begin{equation} \label{eq:gammaj+1}
\gamma_{j+1}=\ell^{(q)}\odot S^q\gamma_j=\ell^{(q)}\odot S^{(j+1)q}\bigodot_{i=0}^{j-1}S^{(i+1)q}\ell^{(q)}=S^{(j+1)q}\bigodot_{i=0}^{j}S^{iq}\ell^{(q)}.
\end{equation}
This proves the first equality in~\eqref{eq:gamma-explicit}
by induction. Substituting \eqref{eq:ellq} into the product on the
right-hand side gives the second equality.
\end{IEEEproof}

\subsection{Polynomial composition of functions and quasi-conservation of related landscapes}
\label{subsec:finite-pcp}

Kriepke~\cite[Theorem~3.9]{kriepke2026shift} stated without proof the equivalence of quasi-conservation and polynomial composition on $\mathbb F_2^{\mathbb Z}$.
In this subsection, we show this equivalence directly on $\F_2^n$  under a suitable condition on $n$. Recall that $\ell$ is a landscape function defined in (\ref{eq:landscape-function}), and $k$ and $m$ are the smallest nonnegative integers
such that $R,T\subseteq[-k,m]$. Throughout this subsection, we assume
\begin{equation}\label{eq:finite-pcp-bound}
n>\max\{2m+k,2k+m\}.
\end{equation}

\begin{theorem}\label{thm:equivalence}
Let $\ell$ be the landscape function in~\eqref{eq:landscape-function}, and let $\{\gamma_j\}_{j\ge0}$ be the sequence defined in~\eqref{eq:gammaj}. If~\eqref{eq:finite-pcp-bound} holds, then the following statements are equivalent:
\begin{enumerate}
\item[(i)] The landscape $\ell$ is quasi-conserved with special index $q$.
\item[(ii)] The sequence $\{\gamma_j\}_{j\ge0}$ satisfies the polynomial composition property.
\end{enumerate}
\end{theorem}

To prove Theorem~\ref{thm:equivalence} we first establish the following two lemmas.

\begin{lemma}\label{lem:two}
Let $\ell$ be a quasi-conserved landscape defined in~\eqref{eq:landscape-function} with special index $q$. Let $a,b$ be positive integers and $u, v \in \mathbb Z$. If $v-u\in\supp\setminus\{q\}$ or $u-v\in\supp\setminus\{q\}$, then
\begin{equation}\label{eq:suvgammaab}
S^u\gamma_a\odot S^v\gamma_b=\mathbf 0.
\end{equation}
\end{lemma}

\begin{IEEEproof}
We only show the case $v-u\in\supp\setminus\{q\}$, and the other case
can be proved similarly. Applying $S^{-u}$ to the left-hand side of
(\ref{eq:suvgammaab}) and setting $\delta=v-u$, it suffices to prove
\begin{equation}\label{eq:gammadelta}
	\gamma_a\odot S^\delta\gamma_b=\mathbf0.
\end{equation}
Here $\delta\in\supp\setminus\{q\}$. Next, we only need to show that
the equality (\ref{eq:gammadelta}) holds.

Since $\ell$ is quasi-conserved, by Lemma~\ref{lem:quasi-cons-set} there exist $r\in R$ and $t\in T$ satisfying $r+\delta\equiv t\pmod n$ or $t+\delta\equiv r\pmod n$. Since $r,t,\delta\in\supp\subseteq[-k,m]$, we have $|r+\delta-t|,|t+\delta-r|\le\max\{2m+k,2k+m\}<n$ by~\eqref{eq:finite-pcp-bound}. Hence $r+\delta=t$ or $t+\delta=r$. The following discussion is divided into two cases.
	
{\bf Case 1:} $t \neq q$.  If $r+\delta=t$, then $\ell^{(q)}$ contains the factor $S^t$, while $S^\delta\ell^{(q)}$ contains the factor $S^{\delta+r}+\mathbf1=S^t+\mathbf 1$, since $t\in T\setminus\{q\}$ and $r\in R$. By~\eqref{eq:gamma-explicit}, these factors also occur in $\gamma_a$ and $S^\delta\gamma_b$, respectively. So, the equality (\ref{eq:gammadelta}) holds. If $t+\delta=r$, then by similar analysis we know that $\gamma_a$ and $S^\delta\gamma_b$ contain the factors $S^r+\mathbf1$ and $S^r$, respectively. Hence, the equality (\ref{eq:gammadelta}) also holds.

{\bf Case 2:}  $t=q$. If $r+\delta=q$, then $r=q-\delta\in\supp\setminus\{q\}$. For $a=1$, $\gamma_1=\ell$ contains $S^q$, while $S^\delta\gamma_b$ contains $S^{\delta+r}+\mathbf1=S^q+\mathbf1$. Hence the left hand side of (\ref{eq:gammadelta}) contains both $S^q$ and $S^q+\mathbf1$, and so the equality (\ref{eq:gammadelta}) holds.
Next, we prove the equality (\ref{eq:gammadelta}) by induction on $a+b$. Assume that
\[\gamma_{\bar a}\odot S^\delta\gamma_{\bar b}=\mathbf0 \]
for all $\bar a,\bar b\ge1$ with $\bar a+\bar b<a+b$ and every $\delta\in\supp\setminus\{ q \}$.

For $a\ge2$,  the recursion \eqref{eq:gammaj} gives $\gamma_a=\ell^{(q)}\odot S^q\gamma_{a-1}$. Since $S^q=S^\delta S^{q-\delta}=S^\delta S^r$ and $S^\delta(F\odot G)=S^\delta F\odot S^\delta G$, we obtain
\[ \gamma_a\odot S^\delta\gamma_b =\ell^{(q)}\odot S^q\gamma_{a-1}\odot S^\delta\gamma_b =\ell^{(q)}\odot S^\delta\bigl(\gamma_b\odot S^r\gamma_{a-1}\bigr) = \mathbf 0, \]
where the last equality follows from the induction hypothesis with $(\bar a,\bar b)=(b,a-1)$ and $r=q-\delta\in\supp\setminus\{q\}$.

If $q+\delta=r$, where $r\in R\subseteq\supp\setminus\{q\}$, we prove the equality (\ref{eq:gammadelta}) by induction on $a+b$. For $b=1$, we have $S^\delta\gamma_1=S^\delta\ell^{(q)}\odot S^r$, while $\gamma_a$ contains the factor $S^r+\mathbf1$ by \eqref{eq:gamma-explicit}. Thus $\gamma_a\odot S^\delta\gamma_1=\mathbf0$. For $b\ge2$, the recursion \eqref{eq:gammaj} and the induction hypothesis with $(\bar a,\bar b)=(a,b-1)$ give
\[
\gamma_a\odot S^\delta\gamma_b
=S^\delta\ell^{(q)}\odot
\bigl(\gamma_a\odot S^r\gamma_{b-1}\bigr)
=\mathbf0.
\]
This completes the induction and proves \eqref{eq:gammadelta}.
\end{IEEEproof}
Lemma~\ref{lem:two} relates the functions $\gamma_a$ and $\gamma_b$ for positive integers $a$ and $b$. To prove Theorem~\ref{thm:equivalence}, we need to show that certain terms in a product expansion are zero.
For this purpose, we introduce
the following notation: for $\lambda\in\supp$ define
\[\zeta_\lambda= \begin{cases}
S^\lambda+\mathbf1,&\lambda\in R,\\
S^\lambda,&\lambda\in T.
\end{cases}
\]
Let $E$ be a nonempty finite subset of $\mathbb Z_{\ge 1}$ and $V\subseteq\supp\setminus\{q\}$. Define
\begin{equation}\label{eq:PVE}
P_{V,E} =\bigodot_{\sigma\in V} \left(\sum_{i\in E}S^\sigma\gamma_i\right) \bigodot_{\lambda\in\supp\setminus(V\cup\{q\})}\zeta_\lambda.
\end{equation}

\begin{lemma}\label{lem:combined}
Let $\ell$ be a quasi-conserved landscape in~\eqref{eq:landscape-function} with special index $q$. With the notation introduced above, if $V\subseteq\supp\setminus\{q\}$ is nonempty, then
\[ P_{V,E}\odot S^q\gamma_c=\mathbf 0 \]
for every $c\ge0$.
\end{lemma}

\begin{IEEEproof}
We prove the lemma by contradiction.  Assume  $P_{V,E}\odot S^q\gamma_c\ne\mathbf 0$ for some $c\ge 0$. Then there  exist $x\in\mathbb F_2^n$ and $j\in[0,n-1]$ such that $(P_{V,E}\odot S^q\gamma_c)(x)_j=1$. So,
$P_{V,E}(x)_j=S^q\gamma_c(x)_j=1$, and every factor of $P_{V,E}$ takes value $1$ at the $j$-th coordinate.
	
We first claim that, for every $a\in\supp\setminus\{q\}$ and  $k\geq 1$, if $S^a \gamma_k(x)_j=1$, then there exists a $k'\ge 1$ such that
\begin{equation}\label{eq:combined-step}
a+q\in R \quad\text{and}\quad S^{a+q}\gamma_{k'}(x)_j=1.
\end{equation}
Fix such $a$ and $k$. Since $\ell$ is quasi-conserved, Lemma~\ref{lem:quasi-cons-set} and the condition \eqref{eq:finite-pcp-bound} imply that there  exist $r\in R$ and $t\in T$ such that \(r+a=t\) or \(t+a=r\). The following discussion is divided into two cases.
	
\textbf{Case 1:} $r+a=t$. If $t\in V$, then from $P_{V,E}(x)_j=1$, we have that $(\sum_{i\in E}S^t\gamma_i(x))_j=1$. So, there exists an $i\in E$ such that $(S^t\gamma_i(x))_j=1$. Since $t-a=r\in\supp\setminus\{q\}$, from the equality~\eqref{eq:suvgammaab}
we have $S^a\gamma_k\odot S^t\gamma_i=\mathbf 0$. This contradicts $S^a\gamma_k(x)_j=S^t\gamma_i(x)_j=1$. If $t=q$ and $c\ge 1$, then $q-a=r\in\supp\setminus\{q\}$, and the equality~\eqref{eq:suvgammaab} implies that $S^a\gamma_k\odot S^q\gamma_c=\mathbf0$. This contradicts $S^a\gamma_k(x)_j=S^q\gamma_c(x)_j=1$. If $t=q$ and $c=0$, then $1= (S^q \gamma_0 (x))_j = S^q(x)_j$, while $S^{a+r}+\mathbf1=S^q+\mathbf1$ is a factor of $S^a\ell^{(q)}$, and so of $S^a\gamma_k$. Thus $(S^q+\mathbf1)(x)_j=1$, forcing $x_{j+q}=0$, which contradicts $S^q(x)_j=x_{j+q}=1$. Finally, if $t\notin V\cup\{q\}$, from $P_{V,E}(x)_j=1$ we have $\zeta_t(x)_j = S^t(x)_j=1$, while $S^{a+r}+\mathbf1=S^t+\mathbf1$ is a factor of $S^a\ell^{(q)}$, and hence a factor of $S^a\gamma_k$. Thus, $(S^t+\mathbf1)(x)_j=1$, forcing $x_{j+t}=0$, which contradicts  $S^t(x)_j=x_{j+t}=1$.
	
\textbf{Case 2:} $t+a=r$. First, suppose that $t\ne q$. If $r\in V$, then $(\sum_{i\in E}S^r\gamma_i(x))_j=1$, and so there exists an $i\in E$ such that $S^r\gamma_i(x)_j=1$. Combining this with the known assumption $S^a \gamma_k(x)_j=1$, we obtain a contradiction to the equality~\eqref{eq:suvgammaab} since $r-a=t\in\supp\setminus\{q\}$. If $r\notin V$, then $P_{V,E}(x)_j=1$ implies that $ \zeta_r(x)_j = (S^r+\mathbf1)(x)_j=1$. On the other hand, $S^{a+t}=S^r$ is a factor of $S^a\ell^{(q)}$, and so
a factor of $S^a\gamma_k$. Thus $S^r(x)_j=1$. This is impossible. Therefore, $t=q$ and $a+q=r\in R$.
	
It remains to find $k'\ge1$ such that $S^{a+q}\gamma_{k'}(x)_j=1$. Recall that $k\ge1$ and $S^a\gamma_k(x)_j=1$ by the assumption. If $k=1$, then $S^a\gamma_1=S^{a+q}\odot S^a\ell^{(q)}$. Since $S^a\gamma_1(x)_j=1$, we have $S^{a+q}(x)_j=x_{j+a+q}=x_{j+r}=1$. This fact implies that $r\in V$. Otherwise, $P_{V, E}$ contains a factor $\zeta_r=S^r+\mathbf 1$, and  $P_{V, E}(x)_j=1$, and so $x_{j+r}=0$. This is a contradiction. Since $r\in V$, we have $\sum_{i\in E}S^r\gamma_i(x)_j=1$. So there exists $k'\in E$ such that $S^r\gamma_{k'}(x)_j=S^{a+q}\gamma_{k'}(x)_j=1$. If $ k \ge2$, then $S^a\gamma_k=S^a\ell^{(q)}\odot S^{a+q}\gamma_{k-1}$. From the assumption $S^a\gamma_k(x)_j=1$ we have
$S^{a+q}\gamma_{k-1}(x)_j=1$. Hence, $S^{a+q}\gamma_{k'}(x)_j=1$ holds for $k'=k-1$. Combining Case 1 and Case 2, we complete the proof of the claim.

Since $V\ne\varnothing$, choose $v\in V$. From $P_{V,E}(x)_j=1$
we have $(\sum_{i\in E}S^v\gamma_i(x))_j=1$, so there exists an
$e_1\in E$ such that $S^v\gamma_{e_1}(x)_j=1$. The claim implies
that there exists an $e_2\geq1$ such that
\begin{equation*}
	v+q\in R \quad\text{and}\quad
	S^{v+q}\gamma_{e_2}(x)_j=1.
\end{equation*}
The fact $v+q\in R$ implies that $v+q\in\supp\setminus\{q\}$. Applying the claim again, we have $v+2q\in R$ and there exists an $e_3\geq1$ such that
$S^{v+2q}\gamma_{e_3}(x)_j=1$. Repeating this process, we have that
$v+sq\in R\subset\supp$ for every $s\geq1$.
This is impossible since the support set of $\ell$ is finite.  Thus, the original assumption does not hold, and so $P_{V,E}\odot S^q\gamma_c=\mathbf 0 $.
\end{IEEEproof}

After the above two preparatory lemmas, we proceed to prove the main theorem of this section.

\begin{IEEEproof}[\textbf{Proof of Theorem~\ref{thm:equivalence}}]
First, we  derive (ii) from (i). It is easy to see that the equality \eqref{eq:polynomial-composition-property}  holds for $s=0$, or when $\alpha_i=0$ for all $i\in [1, s]$.
Next, let $E=\{ i \in[1,s]:\alpha_i=1\}$ be a nonempty set and 	$\Psi=\gamma_0+\sum_{i\in E}\gamma_i$. Set $\alpha_0=1$. Then  $\Psi=\sum_{i=0}^s\alpha_i\gamma_i$ and
the equality \eqref{eq:polynomial-composition-property} is  rewritten as
\begin{equation}\label{eq:forward-target}
\gamma_j\circ\Psi =\gamma_j+\sum_{i\in E}\gamma_{j+i}.
\end{equation}

In the following, we show the equality (\ref{eq:forward-target}) holds for $j\geq 0$. To this end, we first calculate $\ell^{(q)}\circ\Psi$. For every $\lambda\in\mathbb Z$, we have
$S^\lambda\circ\Psi =S^\lambda+\sum_{i\in E}S^\lambda\gamma_i$. Using this identity and the property of Hadamard product we have
\[\ell^{(q)}\circ\Psi =\bigodot_{r\in R} \left(S^r+\mathbf1+\sum_{i\in E}S^r\gamma_i\right)\bigodot_{t\in T\setminus\{q\}}\left(S^t+\sum_{i\in E}S^t\gamma_i\right)
=\bigodot_{\lambda\in\supp\setminus\{q\}}\left(\zeta_\lambda+\sum_{i\in E}S^\lambda\gamma_i\right)=\sum_{V\subseteq\supp\setminus\{q\}}P_{V,E},\]
where $ P_{V,E}$ is defined in (\ref{eq:PVE}).

We now prove \eqref{eq:forward-target} by induction on $j$. For $j=0$, we have $\gamma_0\circ\Psi =\gamma_0+\sum_{i\in E}\gamma_i$, so the equality \eqref{eq:forward-target} holds. Assume that \eqref{eq:forward-target}
holds for  some $j\ge0$. From~\eqref{eq:gammaj} we  have $\gamma_{j+1} = \ell^{(q)}\odot S^q \gamma_j $, and hence
\begin{equation}\label{eq:thmgammaj+1}
\begin{aligned}
\gamma_{j+1}\circ\Psi&=(\ell^{(q)}\odot S^q\gamma_j)\circ\Psi=(\ell^{(q)}\circ\Psi)\odot S^q(\gamma_j\circ\Psi)=\left(\sum_{V\subseteq\supp\setminus\{q\}}P_{V,E}\right)\odot\left(S^q\gamma_j+\sum_{i\in E}S^q\gamma_{j+i}\right)\\
		&=\sum_{V\subseteq\supp\setminus\{q\}}P_{V,E}\odot S^q\gamma_j+\sum_{i\in E}\sum_{V\subseteq\supp\setminus\{q\}}P_{V,E}\odot S^q\gamma_{j+i}.
\end{aligned}
\end{equation}
For every nonempty $V$, Lemma~\ref{lem:combined} gives $P_{V,E}\odot S^q\gamma_j=\mathbf0$ and $P_{V,E}\odot S^q\gamma_{j+i}=\mathbf0$ for every $i\in E$. So, all terms $P_{V,E}\odot S^q\gamma_j$ and $P_{V,E}\odot S^q\gamma_{j+i}$
in (\ref{eq:thmgammaj+1}) are $\mathbf0$ except for $V=\varnothing$. When $V=\varnothing$, $P_{\varnothing,E} =\bigodot_{\lambda\in\supp\setminus\{q\}}\zeta_\lambda =\ell^{(q)}$. So, the equality (\ref{eq:thmgammaj+1}) is reduced to
\[\gamma_{j+1}\circ\Psi = P_{\varnothing,E}\odot S^q\gamma_j+\sum_{i\in E}P_{\varnothing,E}\odot S^q\gamma_{j+i}=\ell^{(q)}\odot S^q\gamma_j+\sum_{i\in E}\ell^{(q)}\odot S^q\gamma_{j+i}=\gamma_{j+1}+\sum_{i\in E}\gamma_{j+i+1}.\]
This  shows that the equality~\eqref{eq:forward-target}  holds for $j+1$.  This completes the induction and proves (ii).

Next, we derive (i) from (ii). Since (ii) holds, taking $j=s=1$ and $\alpha_1=1$ in \eqref{eq:polynomial-composition-property} gives
\begin{equation}\label{eq:first-pcp-identity}
    	\gamma_1\circ(\gamma_0+\gamma_1)=\gamma_1+\gamma_2.
\end{equation}
Below, we prove (i) holds by contradiction. Suppose that (i) does not hold. By Lemma~\ref{lem:quasi-cons-set}, there exists a $\delta\in\supp\setminus\{q\}$ such that
\begin{equation}\label{eq:delta-separation}
r+\delta\not\equiv t\pmod n,\qquad t+\delta\not\equiv r\pmod n
 \end{equation}
for every $r\in R$ and $t\in T$. Consider the set
\[ \mathcal U =\{u\in\mathbb F_2^n:\ell(u)_0=\ell(u)_\delta=1\}.\]
Next, we show that $\mathcal U\ne\varnothing$. Since $\ell(u)_j=\prod_{r\in R}(1+u_{j+r})\prod_{t\in T}u_{j+t}$ for $u\in \mathbb F_2^n$ we have that $u\in\mathcal U$ if and only if
\begin{equation}\label{eq:U-constraints}
u_r=u_{\delta+r}=0,\qquad u_t=u_{\delta+t}=1
\end{equation}
for every $r\in R$ and $t\in T$. The condition \eqref{eq:finite-pcp-bound} and  $R\cap T=\varnothing$ imply that $r\not\equiv t\pmod n$, and hence $\delta+r\not\equiv\delta+t\pmod n$. Moreover, \eqref{eq:delta-separation} gives $r\not\equiv\delta+t\pmod n$ and $\delta+r\not\equiv t\pmod n$ for all $r\in R$ and $t\in T$.
 Hence no coordinate is forced to both $0$ and $1$ in \eqref{eq:U-constraints}. So there exists a $u\in \mathbb F_2^n$ satisfying (\ref{eq:U-constraints}). Hence $\mathcal U\ne\varnothing$.
The following discussion is divided into two cases.

\textbf{Case 1:} There exists a $u\in\mathcal U$ such that  $\ell(u)_q=0$. Put $y=u+\ell(u)$. By \eqref{eq:U-constraints} we have $y_r=u_r+\ell(u)_r=\ell(u)_r$ for $r\in R$ and $y_t=u_t+\ell(u)_t=1+\ell(u)_t$ for $t\in T$.
Since $u_r=0$ and $u_t=1$, from \eqref{eq:landscape-coordinate} we have
\[(\gamma_1\circ(\gamma_0+\gamma_1))(u)_0=\ell(y)_0=\prod_{r\in R}(1+\ell(u)_r) \prod_{t\in T}(1+\ell(u)_t)=\prod_{a\in\supp}(1+\ell(u)_a)=0,\]
where the last equality follows from $\delta\in\supp$ and $\ell(u)_\delta=1$.

Since $q\in T$ and $u\in\mathcal U$, we have $\ell(u)_0=\ell^{(q)}(u)_0 u_q=1$ by the definitions of $\ell$ and $\ell^{(q)}$, and hence $\ell^{(q)}(u)_0=1$. Recall that $\ell(u)_q=0$. From \eqref{eq:gammaj} we have
\[(\gamma_1+\gamma_2)(u)_0 =\ell(u)_0+(\ell^{(q)}\odot S^q\gamma_1)(u)_0 =\ell(u)_0+\ell^{(q)}(u)_0\ell(u)_q=1. \]
Thus, the two sides of \eqref{eq:first-pcp-identity} take different values at coordinate $0$ for the input $u$. This is a contradiction.

\textbf{Case 2:} Every $u\in\mathcal U$ satisfies $\ell(u)_q=1$. Set $e=q-\delta$. The following discussion is divided into two subcases.

\textbf{Subcase 2.1:} $e\notin T$. We show that  $e\notin\supp$.  Choose $z\in\mathcal U$. Then $\ell(z)_0=\ell(z)_\delta=\ell(z)_q=1$. If $e\in R$, then
\[ 1=\ell(z)_\delta =\prod_{r\in R}(1+z_{\delta+r})\prod_{t\in T}z_{\delta+t}.\]
So $z_q=z_{\delta+e}=0$, whereas $\ell(z)_0=1$ and $q\in T$ imply $z_q=1$. Hence $e\notin R$, and so $e\notin\supp$.

We next construct an input $x\in\mathbb F_2^n$ such that $(\gamma_1\circ(\gamma_0+\gamma_1))(x)_0 \ne(\gamma_1+\gamma_2)(x)_0$. Let $x\in \mathbb F_2^n$ satisfy $x_q=0$ and $x_i=z_i$ for $i\not\equiv q\pmod n$. Since $q\in T$,
$\ell(x)_0= \prod_{r\in R}(1+x_{r})\prod_{t\in T}x_{t} =0$. Moreover, $\ell(z)_0=1$ implies $1=\ell(z)_0=\ell^{(q)}(z)_0z_q$, and hence $\ell^{(q)}(z)_0=1$. Since $x$ and $z$ differ only at the $q$-th coordinate,  $\ell^{(q)}(x)_0=\ell^{(q)}(z)_0=1$.

 We next show $\ell(x)_q=\ell(z)_q=1$ and $\ell(x)_\delta=\ell(z)_\delta=1$. By the definition of $\ell$,
 \[ \ell(x)_q =\prod_{r\in R}(1+x_{q+r})\prod_{t\in T}x_{q+t}, \qquad \ell(z)_q =\prod_{r\in R}(1+z_{q+r})\prod_{t\in T}z_{q+t}.\]
Since $x$ and $z$ differ only at the $q$-th coordinate,  $\ell(x)_q\ne\ell(z)_q$ can occur only if $x_{q+a}\ne z_{q+a}$ for some $a\in\supp$, which requires $q+a\equiv q\pmod n$. This gives $a=0$ by \eqref{eq:finite-pcp-bound}, which is impossible since
$0\notin\supp$. Hence $\ell(x)_q=\ell(z)_q=1$. Similarly, $\ell(x)_\delta\ne\ell(z)_\delta$ can occur only if $x_{\delta+a}\ne z_{\delta+a}$ for some $a\in\supp$, which requires $\delta+a\equiv q\pmod n$. This gives $a=e$ by
\eqref{eq:finite-pcp-bound}, which is impossible since $e\notin\supp$. Therefore,  $\ell(x)_\delta=\ell(z)_\delta=1$.

Put $y=x+\ell(x)$. We have $y_q=x_q+\ell(x)_q=1$. Since $1=\ell^{(q)}(x)_0= \prod_{r\in R}(1+x_{r})\prod_{t\in T\setminus \{q\}} x_{t}$, we have $x_r=0$ for every $r\in R$ and $x_t=1$ for every $t\in T\setminus\{q\}$. Hence,
\[(\gamma_1\circ(\gamma_0+\gamma_1))(x)_0 =\ell(y)_0=y_q\prod_{r\in R}(1+\ell(x)_r)\prod_{t\in T\setminus\{q\}}(1+\ell(x)_t)=\prod_{a\in\supp\setminus\{q\}}(1+\ell(x)_a)=0, \]
where the last equality follows from $\delta\in\supp\setminus\{q\}$ and $\ell(x)_\delta=1$. Moreover,
\[ (\gamma_1+\gamma_2)(x)_0=\ell(x)_0+(\ell^{(q)}\odot S^q\gamma_1)(x)_0=\ell(x)_0+\ell^{(q)}(x)_0\ell(x)_q=1. \]
This contradicts \eqref{eq:first-pcp-identity}.

\textbf{Subcase 2.2:} $e\in T$. Since $\delta\ne0$, we have $e\ne q$. We first show that $q$ and $e$ have opposite signs. Fix an arbitrary $s\in T$. Since
    \[ 1=\ell(u)_q =\prod_{r\in R}(1+u_{q+r})\prod_{t\in T}u_{q+t} \]
for every $u\in\mathcal U$, we have $u_{q+s}=1$ for every $u\in\mathcal U$. By~\eqref{eq:U-constraints}, the coordinate $q+s$ must be forced to $1$. Otherwise, it would be free or forced to $0$, and some $u\in\mathcal U$ would satisfy $u_{q+s}=0$, a contradiction. Hence $q+s\equiv t\pmod n$ or $q+s\equiv\delta+t\pmod n$ for some $t\in T$. Using $e=q-\delta$ and \eqref{eq:finite-pcp-bound}, we have
    \[ s+q=t\in T\quad\text{or}\quad s+e=t\in T.\]
If $q,e>0$, take $s=\max T$. Then $s+q>s$ and $s+e>s$, so $s+q\notin T$ and $s+e\notin T$. This is a contradiction. Similarly, if $q, e<0$, take $s=\min T$. Then $s+q\notin T$ and $s+e\notin T$. This is also a contradiction.
Thus $q$ and $e$ have  opposite signs.

Consider the set
\[\mathcal V =\{u\in\mathbb F_2^n:\ell(u)_0=\ell(u)_e=1\}.\]
By the definition of $\ell$, $u\in\mathcal V$ if and only if
\begin{equation}\label{eq:V-constraints}
    	u_r=u_{e+r}=0,\qquad
    	u_t=u_{e+t}=1
\end{equation}
for every $r\in R$ and $t\in T$.  We first show that $(e+R)\cap T=\varnothing$ and $(e+T)\cap R=\varnothing$. Indeed, if $e+r=t$ for some $r\in R$ and $t\in T$, then $q+r=\delta+e+r=\delta+t$. So for every $u\in\mathcal U$ the factor $1+u_{q+r}=1+u_{\delta+t}$ of $\ell(u)_q$ vanishes since $u_{\delta+t}=1$ by \eqref{eq:U-constraints}. Hence $\ell(u)_q=0$ for every $u\in\mathcal U$, contradicting the assumption of Case~2. Similarly, if $e+t=r$ for some $t\in T$ and $r\in R$, then $q+t=\delta+e+t=\delta+r$, so $u_{q+t}=u_{\delta+r}=0$ for every $u\in\mathcal U$ by \eqref{eq:U-constraints}, and again $\ell(u)_q=0$ for every $u\in\mathcal U$, a contradiction. By an argument similar to the discussion for $\mathcal U \ne\varnothing$, with these two identities in place of \eqref{eq:delta-separation}, we have $\mathcal V\ne\varnothing$.

If some $u\in\mathcal V$ satisfies $\ell(u)_q=0$, then, since $e\in T\setminus\{q\}$ and $\ell(u)_e=1$, by  an argument similar to that in Case~1 with $e$ in place of $\delta$ we get a contradiction to \eqref{eq:first-pcp-identity}.
Hence $\ell(u)_q=1$ for every $u\in\mathcal V$.

Fix an arbitrary $s\in T$. Since $\ell(u)_q=1$ for every $u\in\mathcal V$, we have $u_{q+s}=1$ for every $u\in\mathcal V$. By \eqref{eq:V-constraints}, this can hold for every $u\in\mathcal V$ only if  the coordinate $q+s$ is forced to $1$ by \eqref{eq:V-constraints}, i.e., $q+s\equiv t\pmod n$ or $q+s\equiv e+t\pmod n$
for some $t\in T$. Using $\delta=q-e$ and
\eqref{eq:finite-pcp-bound}, we obtain
\[s+q=t\in T \quad\text{or}\quad s+\delta=t\in T. \]
If $q,\delta>0$, take $s=\max T$. Then $s+q\notin T$ and $s+\delta\notin T$. This is a contradiction. Similarly, if $q,\delta<0$, take $s=\min T$ and obtain the same contradiction. Thus $q$ and $\delta$ have opposite signs.

Hence, $e$ and $\delta$ have the same sign, so $e+\delta$ has sign opposite to $q$. This  contradicts $q=e+\delta$. Hence the assumption that (i) does not hold is impossible, and therefore (ii) implies (i). This completes the proof.
\end{IEEEproof}

The bound $n>\max\{2m+k,2k+m\}$ in Theorem~\ref{thm:equivalence} cannot in general be replaced by $n>m+k$.
\begin{example}\label{ex:finite-pcp-small-n}
	Take $n=6$, $R=\{-1,2\}$, $T=\{3\}$, and $q=3$. Here $k=1$, $m=3$, and $n>m+k$. The landscape is $\ell(x)_i=x_{i+3}(1+x_{i-1})(1+x_{i+2})$. With coordinate subscripts taken modulo $6$, we have
	\begin{align*}
	(S^{-1}\ell\odot\ell)(x)_i
	&=x_{i+2}(1+x_{i-2})(1+x_{i+1})
	x_{i+3}(1+x_{i-1})(1+x_{i+2})=0,\\
	(S^2\ell\odot\ell)(x)_i
	&=x_{i+5}(1+x_{i+1})(1+x_{i+4})
	x_{i+3}(1+x_{i-1})(1+x_{i+2})=0.
\end{align*}
Since $\supp\setminus\{q\}=\{-1,2\}$, $\ell$ is quasi-conserved with special index $q=3$ on $\mathbb F_2^6$. By \eqref{eq:gammaj}, we have $\gamma_2(x)_i=x_i(1+x_{i-1})(1+x_{i+2})$. For $x=(0,0,1,1,0,1)$,
we obtain $\gamma_1(x)=\gamma_2(x)=(0,0,1,0,0,1)$, but
\[\gamma_1(x+\gamma_1(x))=(1,0,0,0,0,0)\ne\mathbf0=\gamma_1(x)+\gamma_2(x).\]
Hence, the sequence $\{\gamma_j\}_{j\ge0}$ does not satisfy the polynomial composition property.
\end{example}

\subsection{Long-term behavior of the sequence $\{\gamma_j\}_{j\ge1}$}

We now study the long-term behavior of the sequence $\{\gamma_j\}_{j\ge 1}$. For the rest of this section, assume that $\ell$ is quasi-conserved with special index $q$. By Lemma~\ref{lem:gamma-explicit}, for $j\ge1$ we have
\begin{equation}\label{eq:closed}
\gamma_j=S^{jq}\bigodot_{h=0}^{j-1}\left(\bigodot_{r\in R}(S^{hq+r}+\mathbf1)\bigodot_{t\in T\setminus\{q\}}S^{hq+t}\right).
\end{equation}
Let
\[d=\frac{n}{\gcd(n,q)}, \qquad \mathcal C=\{(r,t)\in R\times T: r\equiv t\pmod{\gcd(n,q)}\}. \]
So, $d$ is the least positive integer such that $dq\equiv 0\pmod n$, and $(r,t)\in\mathcal C$ if and only if $r-t\equiv hq\pmod n$ for some $h\in\mathbb Z$.

\begin{theorem}\label{thm:gamma-behavior}
As $j$ increases, the sequence $\gamma_j$ either becomes zero or remains nonzero and becomes periodic.
\begin{enumerate}
\item[(1)] If $\mathcal C\ne\emptyset$, for each $(r,t)\in\mathcal C$, let
\[j(r,t)=\begin{cases}
		1+\min\{\mu\in[1,d-1]: r-t\equiv\pm\mu q\pmod n\},&t\ne q,\\[2pt]
		\min\{\nu\in[2,d-1]: r\equiv\nu q\pmod n\},&t=q.
	\end{cases}\]
Set $j_0=\min_{(r,t)\in\mathcal C}j(r,t)$. Then $\gamma_j=\mathbf0$ if and only if $j\ge j_0$.
		
\item[(2)] If $\mathcal C=\emptyset$, for every $j\ge1$ set
\[\mathcal Z_j=\bigcup_{h=1}^{j}\left\{(r-hq)\bmod n:r\in R\right\},\qquad \mathcal O_j=\bigcup_{h=1}^{j} \left\{(t-hq)\bmod n:t\in T\setminus\{q\}\right\}. \]
Then there exists a smallest positive integer $j_1$ such that $\mathcal Z_{j_1+1}=\mathcal Z_{j_1}$ and $\mathcal O_{j_1+1}=\mathcal O_{j_1}$, and for every $j\ge j_1$,
\[ \gamma_{j+1}=S^q\gamma_j, \qquad \gamma_{j+d}=\gamma_j, \]
where $d$ is the least positive period of the sequence $\{\gamma_j\}_{j\ge j_1}$.
\end{enumerate}
\end{theorem}

\begin{IEEEproof}
For $j\ge1$, from \eqref{eq:closed}  we know that $\gamma_j=\mathbf 0$ exactly when the expression contains both $S^a$ and $S^a+\mathbf 1$ for some $a\in[0,n-1]$. This occurs exactly when there exist $r\in R$ and $h,h_1,h_2\in[0,j-1]$ such that the first congruence below holds for some $t\in T\setminus\{q\}$ or the second congruence holds:
\begin{equation}\label{eq:gamma-zero-congruences}
		h_1q+r\equiv h_2q+t\pmod n, \qquad hq+r\equiv jq\pmod n.
	\end{equation}
It is known that $\gamma_0$ and $\gamma_1$ are not zero functions. If the first congruence holds then $(r,t)\in\mathcal C$, and the second holds only if $(r,q)\in\mathcal C$. The next discussion is divided into two cases.

(1) The case $\mathcal C\neq\emptyset$. We determine the smallest $j\ge 1$ for which each congruence occurs. If there is a pair $(r,t)\in\mathcal C$ with $t\ne q$ satisfying the first congruence of (\ref{eq:gamma-zero-congruences}), then
	\begin{equation}\label{eq:rtcong}
		r-t\equiv (h_2-h_1)q \pmod n.
	\end{equation}
	From the expression in (\ref{eq:closed}), the smallest $j\ge 1$ satisfying (\ref{eq:rtcong}) is exactly the smallest $|h_2-h_1|$ satisfying (\ref{eq:rtcong}) plus 1. It is known that $dq\equiv0\pmod n$. So, the smallest $j$ satisfying the first congruence is $1+\min\{\mu\in[1,d-1]: r-t\equiv\pm\mu q\pmod n\}$.
	
	If there is a pair $(r,q)\in\mathcal C$ satisfying the second congruence of (\ref{eq:gamma-zero-congruences}), then
	\begin{equation}\label{eq:rqcong}
		r\equiv (j-h)q \pmod n.
	\end{equation}
 It is known that $dq\equiv0\pmod n$ and $j-h\in[1,j]$. So, the smallest $j$ satisfying (\ref{eq:rqcong}) is exactly $\min\{\nu\in[2,d-1]:r\equiv\nu q\pmod n\}$.
	
Combining the two cases, $\gamma_{j_0}=\mathbf0$, and no congruence in \eqref{eq:gamma-zero-congruences} can hold for $j<j_0$. Since $\gamma_{j+1}=\ell^{(q)}\odot S^q\gamma_j$ contains $S^q\gamma_j$ as a factor, $\gamma_j=\mathbf0$ implies $\gamma_{j+1}=\mathbf0$. Hence $\gamma_j=\mathbf0$ if and only if $j\ge j_0$.
	
(2) The case $\mathcal C=\emptyset$. In this case, none of the congruences of \eqref{eq:gamma-zero-congruences} holds. So $\gamma_j\ne\mathbf0$ for every $j\ge1$. We set
	\begin{equation}\label{eq:eta-coordinate}
		\eta_j=S^{-jq}\gamma_j
		=\operatorname{id}\bigodot_{h=1}^{j}S^{-hq}\ell^{(q)}
		\,\,\, {\rm and} \,\,\,
		\eta_j(x)_i
		=x_i\prod_{h=1}^{j}
		\left(\prod_{r\in R}(1+x_{i-hq+r})
		\prod_{t\in T\setminus\{q\}}x_{i-hq+t}\right).
	\end{equation}
	From the definitions of $\mathcal Z_j$ and $\mathcal O_j$, we know that $\eta_j(x)_i=1$ if and only if $x_i=1$, $x_{i+\delta_i}=1$ for every $\delta_i\in\mathcal O_j$, and $x_{i+\delta_i}=0$ for every $\delta_i\in\mathcal Z_j$.
	
	The sets $\{0\}$, $\mathcal Z_j$, and $\mathcal O_j$ are pairwise disjoint. Indeed, if $0\in\mathcal Z_j$, then $r\equiv hq\pmod n$ for some $r\in R$, which gives $(r,q)\in\mathcal C$, a contradiction. If $0\in\mathcal O_j$, then $t\equiv0\pmod{\gcd(n,q)}$ for some $t\in T\setminus\{q\}$. Applying Lemma~\ref{lem:quasi-cons-set} with $\delta=t$, there exist $r\in R$ and $t'\in T$ such that $r+t\equiv t'\pmod n$ or $t'+t\equiv r\pmod n$. Since $t\equiv0\pmod{\gcd(n,q)}$, either congruence gives $r\equiv t'\pmod{\gcd(n,q)}$, so $(r,t')\in\mathcal C$, a contradiction. If $\mathcal Z_j\cap\mathcal O_j\ne\emptyset$, then $r-h_1q\equiv t-h_2q\pmod n$ for some $r\in R$ and $t\in T\setminus\{q\}$, which gives $(r,t)\in\mathcal C$, a contradiction.
	
	We now show that $\eta_{j+1}=\eta_j$ if and only if $\mathcal Z_{j+1}=\mathcal Z_j$ and $\mathcal O_{j+1}=\mathcal O_j$. If $\mathcal Z_{j+1}=\mathcal Z_j$ and $\mathcal O_{j+1}=\mathcal O_j$, from \eqref{eq:eta-coordinate} we have $\eta_{j+1}=\eta_j$. Conversely, suppose there exists a $\delta_z\in\mathcal Z_{j+1}\setminus\mathcal Z_j$. Since $\mathcal Z_{j+1}\cap\mathcal O_{j+1}=\emptyset$ we have that $\delta_z\notin\mathcal O_{j+1}$, and so $\delta_z\notin\mathcal O_j$. From the expression of (\ref{eq:eta-coordinate}) we know that the factors $x_{i+\delta_z}$ and $1+x_{i+\delta_z}$ do not appear in $\eta_j(x)_i$. Hence we may choose $x$ such that $\eta_j(x)_i=1$ and $x_{i+\delta_z}=1$. Since $\delta_z\in\mathcal Z_{j+1}$, $\eta_{j+1}(x)_i$ contains the factor $1+x_{i+\delta_z}=0$. Hence $\eta_{j+1}(x)_i=0$, and thus $\eta_{j+1}\ne\eta_j$. Similarly, if there exists a $\delta_z\in\mathcal O_{j+1}\setminus\mathcal O_j$, we can also derive that $\eta_{j+1}\ne\eta_j$. This shows that $\mathcal Z_{j+1}=\mathcal Z_j$ and $\mathcal O_{j+1}=\mathcal O_j$ from $\eta_{j+1}=\eta_j$.
	
	We next show that if $\mathcal Z_{j+1}=\mathcal Z_j$ for some $j$, then $\mathcal Z_{j+s}=\mathcal Z_j$ for every $s\ge1$. It is known that
	\[
	\mathcal Z_{j+1}
	=\mathcal Z_j\cup
	\left\{(r-(j+1)q)\bmod n:r\in R\right\}.
	\]
	From $\mathcal Z_{j+1}=\mathcal Z_j$, for every $r\in R$, there exist $r^\prime\in R$ and $h\in[1,j]$ such that
	\[
	r-(j+1)q\equiv r^\prime-hq\pmod n.
	\]
	From this congruence we have
	\[
	r-(j+2)q\equiv r^\prime-(h+1)q\pmod n,
	\qquad h+1\in[1,j+1].
	\]
	This congruence shows that $\mathcal Z_{j+2}\subseteq\mathcal Z_{j+1}=\mathcal Z_j$, and then $\mathcal Z_{j+2}=\mathcal Z_j$. By induction on $s$ we have $\mathcal Z_{j+s}=\mathcal Z_j$ for every $s\ge1$. Similarly, $\mathcal O_{j+s}=\mathcal O_j$ for every $s\ge1$ if $\mathcal O_{j+1}=\mathcal O_j$.
	
For every $j\geq1$, $\mathcal Z_j$ and $\mathcal O_j$ are subsets of the set of all integers modulo $n$, and $\mathcal Z_j\subseteq\mathcal Z_{j+1}$ and $\mathcal O_j\subseteq\mathcal O_{j+1}$. So, there must exist a positive integer $j_1$ such that $\mathcal Z_{j_1+1}=\mathcal Z_{j_1}$ and $\mathcal O_{j_1+1}=\mathcal O_{j_1}$. Moreover, assume $j_1$ is the minimal positive integer such that the equalities hold. The discussion above shows $\mathcal Z_j=\mathcal Z_{j_1}$ and $\mathcal O_j=\mathcal O_{j_1}$ for every $j\ge j_1$. From (\ref{eq:eta-coordinate}) we have $\eta_{j+1}=\eta_j$ for every $j\ge j_1$. Since $\eta_j=S^{-jq}\gamma_j$, we get $\gamma_{j+1}=S^q\gamma_j$ for every $j\ge j_1$.
	
Finally, we show that $d$ is the minimal positive integer such that $\gamma_{j+d}=\gamma_j$ for $j\ge j_1$. Since $\gamma_{j+1}=S^q\gamma_j$ for every $j\ge j_1$, we have $\gamma_{j+d}=S^{dq}\gamma_j$ for every $j\ge j_1$. Since $dq\equiv0\pmod n$ and $S^n=\operatorname{id}$, we have $\gamma_{j+d}=\gamma_j$ for every $j\ge j_1$.
	
It remains to show that $d$ is the least positive period of $\{\gamma_j\}_{j\ge j_1}$. If not, there exists $\tau\in[1,d-1]$ such that $\gamma_{j+\tau}=\gamma_j$ for every $j\ge j_1$. Applying $\gamma_{j+1}=S^q\gamma_j$ repeatedly gives $\gamma_{j_1+\tau}=S^{\tau q}\gamma_{j_1}=\gamma_{j_1}$. Together with $\eta_{j_1}=S^{-j_1q}\gamma_{j_1}$ we have
\[S^{\tau q}\eta_{j_1}=S^{-j_1q}S^{\tau q}\gamma_{j_1}=S^{-j_1q}\gamma_{j_1}=\eta_{j_1}.\]
	
We construct an input $x\in \mathbb F_2^n$ satisfying $\eta_{j_1}(x)_i=0$ but $\eta_{j_1}(x)_{i+\tau q}=1$. Every element of $\mathcal O_{j_1}$ is nonzero modulo $\gcd(n,q)$. Indeed, if $\delta_o\in\mathcal O_{j_1}$ satisfies $\delta_o\equiv0\pmod{\gcd(n,q)}$, then $\delta_o\equiv t-hq\pmod n$ for some $t\in T\setminus\{q\}$ and $h\in[1,j_1]$. So $t\equiv0\pmod{\gcd(n,q)}$. Applying Lemma~\ref{lem:quasi-cons-set} with $\delta=t$, there exist $r\in R$ and $t'\in T$ such that $r+t\equiv t'\pmod n$ or $t'+t\equiv r\pmod n$. Since $t\equiv0\pmod{\gcd(n,q)}$, either congruence gives $r\equiv t'\pmod{\gcd(n,q)}$. So $(r,t')\in\mathcal C$ yields a contradiction. Since $-\tau q\equiv0\pmod{\gcd(n,q)}$, we have $-\tau q\notin\mathcal O_{j_1}$.
	
By \eqref{eq:eta-coordinate}, $\eta_{j_1}(x)_{i+\tau q}=1$ exactly when
\[x_{i+\tau q}=1,\qquad x_{i+\tau q+\delta_o}=1,\qquad x_{i+\tau q+\delta_z}=0,\]
for every $\delta_o\in\mathcal O_{j_1}$ and $\delta_z\in\mathcal Z_{j_1}$. We have $\tau q\not\equiv0\pmod n$ since $\tau\in[1,d-1]$ and $d$ is the order of $q$ modulo $n$. Thus $x_{i+\tau q}$ and $x_i$ are distinct coordinate variables. Since $-\tau q\notin\mathcal O_{j_1}$, the coordinate variables $x_i$ and $x_{i+\tau q+\delta_o}$ for $\delta_o\in\mathcal O_{j_1}$ are pairwise distinct. Together with the pairwise disjointness proved above, this allows us to choose $x$ such that $\eta_{j_1}(x)_{i+\tau q}=1$ and $x_i=0$. Then
\[\eta_{j_1}(x)_i=0,\qquad\bigl(S^{\tau q}\eta_{j_1}\bigr)(x)_i=\eta_{j_1}(x)_{i+\tau q}=1.\]
This contradicts $S^{\tau q}\eta_{j_1}=\eta_{j_1}$. Thus, there is no $\tau\in[1,d-1]$ such that $\gamma_{j+\tau}= \gamma_j$ for $j\geq j_1$, and $d$ is the least positive period of $\{\gamma_j\}_{j\ge j_1}$.
	
Moreover, we can show that $j_1$ is also the smallest positive index satisfying $\gamma_{j+d}=\gamma_j$ for $j\geq j_1$. On the contrary, suppose that $\gamma_{j+d}=\gamma_j$ for some positive integer $j<j_1$. Since $dq\equiv0\pmod n$, we have
\[\eta_{j+d}=S^{-(j+d)q}\gamma_{j+d}=S^{-jq}\gamma_j=\eta_j.\]
	
Next, we show that $\mathcal Z_{j+d}=\mathcal Z_j$. Otherwise, by $\mathcal Z_j\cap\mathcal O_j=\emptyset$ and (\ref{eq:eta-coordinate}), we can choose a $\delta_z\in\mathcal Z_{j+d}\setminus\mathcal Z_j$ and an $x$ such that $\eta_j(x)_i=1$ and $x_{i+\delta_z}=1$. Then $\eta_{j+d}(x)_i=0$. This contradicts the fact that $\eta_{j+d}=\eta_j$, and so $\mathcal Z_{j+d}=\mathcal Z_j$. Similarly, $\mathcal O_{j+d}=\mathcal O_j$. On the other hand, we have $\mathcal Z_j\subseteq\mathcal Z_{j+1}\subseteq\mathcal Z_{j+d}$ and $\mathcal O_j\subseteq\mathcal O_{j+1}\subseteq\mathcal O_{j+d}$, and so $\mathcal Z_{j+1}=\mathcal Z_j$ and $\mathcal O_{j+1}=\mathcal O_j$. This contradicts the minimality of $j_1$. So, the integer $j_1$ is the smallest positive index satisfying $\gamma_{j+d}=\gamma_j$ for $j\geq j_1$.
\end{IEEEproof}

The indices $j_0$ and $j_1$ in Theorem~\ref{thm:gamma-behavior} admit the following simple upper bounds. The definition of $j_0$ in Theorem~\ref{thm:gamma-behavior} gives $j_0\le d$ when $\mathcal C\ne\emptyset$.
If $\mathcal C=\emptyset$, then $dq\equiv0\pmod n$ gives $a-(d+1)q\equiv a-q\pmod n$ for every $a\in\supp\setminus\{q\}$. Hence $\mathcal Z_{d+1}=\mathcal Z_d$ and $\mathcal O_{d+1}=\mathcal O_d$, so $j_1\le d$.
Next, we discuss the linear independence of the sequence of functions $\{ \gamma_j \}_{j\geq 0}$.
\begin{proposition}\label{lem:linear-independence}
With the notation of Theorem~\ref{thm:gamma-behavior}, we have
\begin{enumerate}
\item[(1)] If $\mathcal C\ne\emptyset$, then $\gamma_0,\gamma_1,\ldots,\gamma_{j_0-1}$ are linearly independent over $\mathbb F_2$.
\item[(2)] If $\mathcal C=\emptyset$, then $\gamma_0,\gamma_1,\ldots,\gamma_{j_1+d-1}$ are linearly independent over $\mathbb F_2$.
\end{enumerate}
\end{proposition}

\begin{IEEEproof}
(1) The case $\mathcal C\ne\emptyset$. In this case, $\gamma_j =\mathbf 0$ for $j\geq j_0$. Suppose that $\sum_{i=0}^{j_0-1}c_i\gamma_i=\mathbf0$, where $c_0, c_1, \ldots, c_{j_0-1}\in\mathbb F_2$ are not all zero, and let $k$ be the
 smallest index such that $c_k=1$. By the recurrence~\eqref{eq:gammaj},
\[\sum_{i=0}^{j_0-1}c_i\gamma_{i+1}=\sum_{i=0}^{j_0-1}c_i\bigl(\ell^{(q)}\odot S^q\gamma_i\bigr)=\ell^{(q)}\odot S^q\left(\sum_{i=0}^{j_0-1}c_i\gamma_i\right)=\mathbf0. \]
Repeating this process we have $\sum_{i=0}^{j_0-1} c_i\gamma_{i+s}=\mathbf 0$ for every $s\ge0$. Taking $s=j_0-1-k$ and using $c_i=0$ for $i<k$, we obtain
\[\mathbf 0 =\sum_{i=k}^{j_0-1}c_i\gamma_{j_0-1+i-k} =\gamma_{j_0-1}.\]
The last equality holds because $c_k=1$ and $\gamma_j=\mathbf0$ for every $j\ge j_0$. This contradicts $\gamma_{j_0-1}\ne\mathbf0$. Hence $\gamma_0,\ldots,\gamma_{j_0-1}$ are linearly independent over $\mathbb F_2$.
	
(2) The case $\mathcal C=\emptyset$. If $\gamma_0,\gamma_1,\ldots,\gamma_{j_1+d-1}$ are linearly dependent over $\mathbb F_2$, then their coordinate functions $x\mapsto\gamma_0(x)_0,x\mapsto\gamma_1(x)_0,\ldots,x\mapsto\gamma_{j_1+d-1}(x)_0$ are
linearly dependent over $\mathbb F_2$. Next we show that these coordinate functions are linearly independent over $\mathbb F_2$, and so $\gamma_0,\gamma_1,\ldots,\gamma_{j_1+d-1}$ are linearly independent over $\mathbb F_2$.
From $\gamma_0= {\rm id}$ and \eqref{eq:closed} we have
\[ \gamma_0(x)_0 = x_0, \quad \gamma_j(x)_0=x_{jq}\prod_{h=0}^{j-1}\left(\prod_{r\in R}(1+x_{hq+r})\prod_{t\in T\setminus\{q\}}x_{hq+t} \right), \,\, j\geq 1 .\]
For $j\ge1$, the highest-degree monomial of $\gamma_j(x)_0$ is $\prod_{a\in P_j}x_a$. Indeed, every monomial of $\gamma_j(x)_0$ is supported on a subset of $P_j$. By \eqref{eq:gamma-zero-congruences} and $\mathcal C=\emptyset$, the expression for $\gamma_j(x)_0$ does not contain both $x_a$ and $1+x_a$ as factors for any $a$. Hence, for each $a\in P_j$, the factors involving $x_a$ form either $(1+x_a)^w$ or $x_a^w$ with $w\ge1$. In the expansion of $(1+x_a)^w$ with $w\ge1$, exactly $2^w-1$ choices produce $x_a$ after using $x_a^2=x_a$. The number $2^w-1$ is odd, and $x_a^w=x_a$. Hence the coefficient of $\prod_{a\in P_j}x_a$ is $1$, where
\begin{equation}\label{eq:Pj-definition}
P_j=\{jq\bmod n\}\cup\bigcup_{h=0}^{j-1}\{(a+hq)\bmod n:a\in\supp\setminus\{q\}\},
\end{equation}
and we set $P_0=\{0\}$. To prove linear independence of the coordinate functions $\gamma_0(x)_0, \gamma_1(x)_0,\ldots,\gamma_{j_1+d-1}(x)_0$, it suffices to show that $P_0, \ldots, P_{j_1+d-1}$ are pairwise distinct.

We first show that $P_0,\ldots,P_{j_1}$ are pairwise distinct.  For $j\ge1$,
\[\begin{aligned}
\{(a-jq)\bmod n:a\in P_j\}&=\{0\}\cup\bigcup_{h=0}^{j-1}\{(a-(j-h)q)\bmod n:a\in\supp\setminus\{q\}\}\\
	&=\{0\}\cup\bigcup_{h=1}^{j}\{(a-hq)\bmod n:a\in\supp\setminus\{q\}\}\\
	&=\{0\}\cup\mathcal Z_j\cup\mathcal O_j,
\end{aligned}\]
where the three sets $\{0\}$, $\mathcal Z_j$ and $\mathcal O_j$  are proved to be pairwise disjoint in Theorem~\ref{thm:gamma-behavior}.
Since $|P_1|=|R\cup T|$, we have $|P_1|\ge2>|P_0|$. For $1\le j<j_1$, we have $\mathcal Z_j\subsetneq\mathcal Z_{j+1}$ or $\mathcal O_j\subsetneq\mathcal O_{j+1}$.
Since $j_1$ is the first index $j$ for which $\mathcal Z_j=\mathcal Z_{j+1}$ and $\mathcal O_j=\mathcal O_{j+1}$, we have
\[|P_0|<|P_1|<\cdots<|P_{j_1}|.\]
	
Next, we show that $\delta \not\equiv0\pmod{\gcd(n,q)}$ for every $\delta\in\supp\setminus\{q\}$. Suppose instead that $\delta\equiv0\pmod{\gcd(n,q)}$ for some such $\delta$. Then Lemma~\ref{lem:quasi-cons-set} gives $r\in R$ and $t\in T$
such that $r+\delta\equiv t\pmod n$ or $t+\delta\equiv r\pmod n$. Thus $r\equiv t\pmod{\gcd(n,q)}$, and so $(r,t)\in\mathcal C$. This contradicts that $\mathcal C=\emptyset$. Therefore, for every $\delta \in\supp\setminus\{q\}$, we have
\[(\delta+hq)\bmod n\not\equiv0\pmod{\gcd(n,q)}, \,\,\,  h\in[0,j-1],\]
where $j\ge 1$. Thus by \eqref{eq:Pj-definition}, $jq\bmod n$ is the unique element of $P_j$ congruent to $0$ modulo $\gcd(n,q)$ for every $j\ge1$.
	
We now show that $P_{j_1},\ldots,P_{j_1+d-1}$ are pairwise distinct. Suppose that $P_{j_1+s_1}=P_{j_1+s_2}$ for $0\le s_1<s_2<d$. Then $((j_1+s_1)q\bmod n)\in P_{j_1+s_2}$. Since $(j_1+s_1)q\equiv0\pmod{\gcd(n,q)}$ and
the only element of $P_{j_1+s_2}$ congruent to $0$ modulo $\gcd(n,q)$ is $(j_1+s_2)q\bmod n$, we have
\[(j_1+s_1)q\equiv(j_1+s_2)q\pmod n.\]
Hence $(s_2-s_1)q\equiv0\pmod n$, which is impossible because $0<s_2-s_1<d$ and $d$ is the least positive integer satisfying $dq\equiv0\pmod n$. Therefore $P_{j_1},\ldots,P_{j_1+d-1}$ are pairwise distinct.
	
For $j\ge j_1$, both $\mathcal Z_j$ and $\mathcal O_j$ remain constant, so $|P_j|=|P_{j_1}|$. Since $|P_0|<|P_1|<\cdots<|P_{j_1}|$ and $P_{j_1},\ldots,P_{j_1+d-1}$ are pairwise distinct, it follows that $P_0,\ldots,P_{j_1+d-1}$
are pairwise distinct. This completes the proof.
\end{IEEEproof}

\subsection{Shift-invariant permutations and their algebraic structure}

In this subsection, we use the linearly independent functions $\gamma_0, \gamma_1, \ldots, \gamma_{\deg p(z) -1}$ to construct a group under composition of functions, which is isomorphic to the multiplicative group of units of a polynomial quotient ring, where $p(z)$ is defined below. The functions in the group are shift-invariant permutations of
$\mathbb F_2^n$, and their inverses, iterates and composition orders are calculated in the polynomial quotient ring.

Let $\ell$ be a quasi-conserved landscape function over $\mathbb F_2^n$ with special index $q$. Let $\{\gamma_j\}_{j\geq 0}$ be  the sequence of functions defined in (\ref{eq:gammaj}), and let $j_0$, $j_1$ and $d$ be as in Theorem~\ref{thm:gamma-behavior}.
In the following we define
\begin{equation}\label{eq:quotient-polynomial}
p(z)=\begin{cases}
		z^{j_0},&\mathcal C\ne\emptyset,\\[2pt]
		z^{j_1}(1+z^d),&\mathcal C=\emptyset.
	\end{cases}
\end{equation}

Let $\Gamma=\operatorname{span}\{\gamma_0,\gamma_1,\ldots\}$ be a linear space over $\mathbb F_2$ and $\mathcal Q=\mathbb F_2[z]/\langle p(z)\rangle$. By Theorem~\ref{thm:gamma-behavior} and Proposition~\ref{lem:linear-independence}, $\gamma_0,\ldots,\gamma_{\deg p(z)-1}$ form a basis of $\Gamma$.
For $F=\sum_{i=0}^{\deg p(z)-1}a_i\gamma_i$, define a  vector-space isomorphism $\varphi:\Gamma\to\mathcal Q$ by
\[\varphi(F)=\sum_{i=0}^{\deg p(z)-1}a_i z^i.\]
To construct a group of shift-invariant permutations of $\mathbb F_2^n$, we define
\[\mathcal G =\gamma_0+\operatorname{span}\left\{ \gamma_1,\ldots,\gamma_{\deg p(z)-1}\right\}\,\, \,{\rm and}\,\,\, \mathcal M=\left\{1+\sum_{i=1}^{\deg p(z)-1}a_i z^i:a_i\in\mathbb F_2\right\}\subseteq\mathcal Q.\]
It is easy to see that $\mathcal G$ and $\mathcal{M}$ are monoids. Let $\mathcal G^*$ and $\mathcal{M}^*$ be their groups of units, respectively.
By Theorem~\ref{thm:equivalence}, the sequence $\{\gamma_j\}_{j\ge 0}$ satisfies the polynomial composition property, which allows us to have the following isomorphisms.

\begin{theorem}\label{thm:main}
The set $\mathcal G$ is an abelian monoid under composition. The restriction of $\varphi$ to $\mathcal G$ is a monoid isomorphism $(\mathcal G,\circ)\cong(\mathcal M,\cdot)$,
and its restriction to $\mathcal G^*$ is a group isomorphism $(\mathcal G^*,\circ)\cong(\mathcal M^*,\cdot)$. Moreover, $\mathcal G^*=\mathcal G$ if and only if $\mathcal C\ne\emptyset$.
\end{theorem}

\begin{IEEEproof}
Let $F,G\in\mathcal G$ and write $F=\sum_i a_i\gamma_i$ and $G=\sum_j b_j\gamma_j$, where $a_i, b_i\in \mathbb F_2$ and $a_0=b_0=1$. By the polynomial composition property, $\gamma_i\circ G=\sum_j b_j\gamma_{i+j}$ for every $i$. Hence
\[F\circ G=\sum_i a_i(\gamma_i\circ G)=\sum_{i,j}a_i b_j\gamma_{i+j}.\]
By this equality and the definition of $\varphi$, we have
\[\varphi(F\circ G)=\varphi(F)\varphi(G)\in\mathcal M.\]
Here the reduction of $\varphi(F)\varphi(G)$ modulo $p(z)$ is compatible with the relations in $\Gamma$: $\gamma_j=\mathbf0$ for $j\ge j_0$ corresponds to $z^j\equiv0\pmod{z^{j_0}}$, and $\gamma_{j+d}=\gamma_j$ for $j\ge j_1$ corresponds to $z^{j+d}\equiv z^j\pmod{z^{j_1}(1+z^d)}$. This shows that $\varphi$ is a homomorphism from $\mathcal G$ to $\mathcal{M}$. Moreover, it is easy to see that $\varphi|_{\mathcal G}:(\mathcal G,\circ)\to(\mathcal M,\cdot)$ is a monoid isomorphism,
and so $\varphi|_{\mathcal G^*}:(\mathcal G^*,\circ)\to(\mathcal M^*,\cdot)$ is a group isomorphism.
	
If $\mathcal C\ne\emptyset$, every element of $\mathcal M$ is a unit modulo $z^{j_0}$. If $\mathcal C=\emptyset$, $1+z\in\mathcal M$ is not a unit since $(1+z)\mid p(z)$. Hence $\mathcal G^*=\mathcal G$ if and only if $\mathcal C\ne\emptyset$.
\end{IEEEproof}

When $\mathcal C=\emptyset$, Theorem~\ref{thm:main} yields the following permutation criterion.

\begin{corollary}\label{cor:pp-period}
	Suppose $\mathcal C=\emptyset$, and let
	$F=\operatorname{id}+\sum_{j=1}^{j_1+d-1}a_j\gamma_j\in\mathcal G$.
	Then $F$ is a permutation of $\mathbb F_2^n$ if and only if
	\[\gcd\left(1+\sum_{j=1}^{j_1+d-1}a_jz^j,\,1+z^{d_1}\right)=1,\]
	where $d_1$ is the largest odd divisor of $d$.
\end{corollary}

\begin{IEEEproof}
By Theorem~\ref{thm:main}, $F$ is a permutation if and only if $\varphi(F)$ is a unit in $\mathcal Q$. This holds if and only if the polynomial $\varphi(F)=1+\sum_{j=1}^{j_1+d-1}a_jz^j$ is coprime to $p(z)=z^{j_1}(1+z^d)$. Since $\varphi(F)$ is coprime to $z^{j_1}$, it remains to check whether it is coprime to $1+z^d$. Writing $d=2^s d_1$ with $s\ge0$ gives $1+z^d=(1+z^{d_1})^{2^s}$. So, $\varphi(F)$ is coprime to $p(z)$ if and only if $\gcd( \varphi(F), 1+ z^{d_1}) =1$.
\end{IEEEproof}

Next, we determine the inverses, iterates and orders of functions in the group $\mathcal{G}^*$.
\begin{theorem}\label{thm:algebraic-properties}
Let $F=\operatorname{id}+\sum_{j=1}^{u}a_j\gamma_j\in\mathcal G^*$, where $u=j_0-1$ if $\mathcal C\ne\emptyset$ and $u=j_1+d-1$ if $\mathcal C=\emptyset$.
\begin{enumerate}
\item[(1)] The inverse of $F$ has the form $F^{-1}=\operatorname{id}+\sum_{k=1}^{u}b_k\gamma_k$, and the coefficients $b_k$ are determined by the following: \\
If $\mathcal C\ne\emptyset$, then $b_k=a_k+\sum\limits_{\substack{i+j=k\\i,j\in[1,u]}}a_i b_j$ for $k\in[1,u]$.
If $\mathcal C=\emptyset$, then $b_k=\begin{cases}
			a_k+\displaystyle\sum_{\substack{i+j=k\\i,j\in[1,u]}}
			a_i b_j,& \hspace{-0.1cm} k\in[1,j_1-1],\\
			a_k+\displaystyle\sum_{\substack{i+j\equiv k\pmod d\\
			i+j\ge j_1,\ i,j\in[1,u]}}a_i b_j,
			&k\in[j_1,u].
		\end{cases}
		$
\item[(2)] Set $a_0=1$. For $s\ge1$, the $s$th iterate of $F$ has the form $F^s=\operatorname{id}+\sum_{k=1}^{u}c_k\gamma_k$, and the coefficients $c_k$ are determined by the following:\\
 If $\mathcal C\ne\emptyset$, then  $c_k=\sum\limits_{\substack{i_1+\cdots+i_s=k\\
		i_1,\ldots,i_s\in[0,u]}}
		a_{i_1}\cdots a_{i_s}$ for $k\in[1,u]$.
If $\mathcal C=\emptyset$, then $c_k=\begin{cases}
			\displaystyle\sum_{\substack{i_1+\cdots+i_s=k\\
					i_1,\ldots,i_s\in[0,u]}}
			a_{i_1}\cdots a_{i_s}, &  \hspace{-0.6cm} k\in[1,j_1-1], \\
			\displaystyle\sum_{\substack{i_1+\cdots+i_s\equiv k\pmod d\\
					i_1+\cdots+i_s\ge j_1\\
					i_1,\ldots,i_s\in[0,u]}}
			       a_{i_1}\cdots a_{i_s},
			& k \in[j_1,u].
		\end{cases}
		$
\item[(3)] Let $\widehat j=\min\{j\in[1,u]:a_j=1\}$ and $F\ne\operatorname{id}$. If $\mathcal C\ne\emptyset$, then
	\[\operatorname{ord}(F)=2^{\lceil\log_2(j_0/\widehat j)\rceil}.\]
If $\mathcal C=\emptyset$ and $d$ is a power of two, write $\sum_{j=1}^{u}a_jz^j=(1+z)^\rho A_1(z)$ with $\rho\ge1$ and $A_1(1)=1$. Then
		\[\operatorname{ord}(F)=2^{\left\lceil \log_2\max\left\{1,\frac{j_1}{\widehat j},\frac d\rho\right\}\right\rceil}. \]
\end{enumerate}
\end{theorem}

\begin{IEEEproof}
Let $A(z)=1+\sum_{j=1}^{u}a_jz^j$. Then $\varphi(F)=A(z)$.
	
(1) Let $B(z)=1+\sum_{k=1}^{u}b_kz^k$. Since $A(z)$ is a unit in $\mathcal Q$, $A(z)B(z)\equiv1\pmod{p(z)}$ uniquely determines the coefficients $b_k$. For $\mathcal C\ne\emptyset$, we have $p(z)=z^{j_0}$. The terms of degree at least
 $j_0$ vanish modulo $p(z)$, so comparing the coefficient of $z^k$ for $k\in[1,u]$ gives
\[a_k+b_k+\sum_{\substack{i+j=k\\i,j\in[1,u]}}a_i b_j=0.\]
This yields the first formula.

For $\mathcal C=\emptyset$, it is known that $p(z)=z^{j_1}(1+z^d)$, and so $z^{N+d}\equiv z^N \pmod{p(z)}$ for $N\ge j_1$. We keep the terms of degree less than $j_1$ and replace each $z^N$ with $N\ge j_1$ by $z^k$, where $k\in[j_1,u]$ is uniquely determined by $N\equiv k\pmod d$. Comparing coefficients then gives the remaining equations in (1).

(2) Set $a_0=1$. By Theorem~\ref{thm:main}, $\varphi(F^s)=A(z)^s$. Expanding the product gives
\[A(z)^s=\left( \sum_{i=0}^{u} a_{i}z^{i}\right)^s =\sum_{i_1,\ldots,i_s\in[0,u]} a_{i_1}\cdots a_{i_s}z^{i_1+\cdots+i_s}.\]
For $\mathcal C\ne\emptyset$, we have $p(z)=z^{j_0}$, so terms of total degree at least $j_0$ vanish modulo $p(z)$. Hence comparing the coefficient of $z^k$ for $k\in[1,u]$ gives
\[c_k=\sum_{\substack{i_1+\cdots+i_s=k\\
			i_1,\ldots,i_s\in[0,u]}}
	a_{i_1}\cdots a_{i_s}.
\]
For $\mathcal C=\emptyset$, we have $z^{N+d}\equiv z^N\pmod{p(z)}$ for $N\ge j_1$, where $p(z)=z^{j_1}(1+z^d)$. Under this reduction, terms of total degree less than $j_1$ remain unchanged, while each monomial $z^{i_1+\cdots+i_s}$ of total degree at least $j_1$ reduces to $z^k$ for the unique $k\in[j_1,u]$ satisfying $i_1+\cdots+i_s\equiv k\pmod d$. Comparing coefficients after reduction, we have the remaining formula in~(2).
	
(3) Let $F\ne\operatorname{id}$. For every $\nu\ge0$, Theorem~\ref{thm:main} shows that $F^{2^\nu}=\operatorname{id}$ if and only if $A(z)^{2^\nu}\equiv1\pmod{p(z)}$, or equivalently $p(z)\mid A(z)^{2^\nu}+1$. It is easy to see that
\[A(z)^{2^\nu}+1=\left(\sum_{j=1}^{u}a_jz^j\right)^{2^\nu}=\sum_{j=1}^{u}a_jz^{j2^\nu}.\]
	
For $\mathcal C\ne\emptyset$, we have $p(z)=z^{j_0}$. The lowest-degree nonzero term of this sum is $z^{\widehat j2^\nu}$. So $F^{2^\nu}=\operatorname{id}$ if and only if $\widehat j2^\nu\ge j_0$.  Since $|\mathcal M|=2^{j_0-1}$ and every element of $\mathcal M$ is a unit here, $\mathcal M^*$ is a $2$-group, so the order of $F$ is indeed a power of $2$. The least such $\nu$ is $\lceil\log_2(j_0/\widehat j)\rceil$, which gives the first formula.
	
For $\mathcal C=\emptyset$ with $d$ a power of two, we have $p(z)=z^{j_1}(1+z)^d$. Since $A(z)$ is a unit in $\mathcal Q$ and $F\ne\operatorname{id}$, we have $A(1)=1$ and $A(z)+1\ne 0$. Thus we can write $A(z)+1=(1+z)^\rho A_1(z)$ with $\rho\ge1$ and $A_1(1)=1$. Consequently, $A(z)^{2^\nu}+1=(1+z)^{\rho2^\nu}A_1(z)^{2^\nu}$.

Finally, we determine when $p(z)$ divides $A(z)^{2^\nu}+1$. Since $A_1(1)=1$, the multiplicity of $1+z$ in $A(z)^{2^\nu}+1$ is exactly $\rho2^\nu$. Thus $(1+z)^d$ divides $A(z)^{2^\nu}+1$ if and only if $\rho2^\nu\ge d$. By the definition of $\widehat j$, the lowest-degree term of $A(z)^{2^\nu}+1$ is $z^{\widehat j2^\nu}$. So $z^{j_1}$ divides $A(z)^{2^\nu}+1$ if and only if $\widehat j2^\nu\ge j_1$. Hence $F^{2^\nu}=\operatorname{id}$ if and only if $\widehat j2^\nu\ge j_1$ and $\rho2^\nu\ge d$. Both inequalities hold for sufficiently large $\nu$, so the order of $F$ is a power of $2$. Taking the least $2^\nu$ satisfying both inequalities gives the second formula.
	
\end{IEEEproof}

\section{Applications: Permutation Constructions via Binomials and Trinomials in $\mathcal Q$}
\label{sec:families}

In this section, we use the results in Section~\ref{sec3} to construct explicit shift-invariant permutations of $\mathbb F_2^n$. Recall that $\ell$ is the quasi-conserved landscape function defined in (\ref{eq:landscape-function}) with special index $q$
and $\{\gamma_j\}_{j\ge 0}$ is the sequence of functions defined in (\ref{eq:gammaj}). By Theorem~\ref{thm:equivalence}, this sequence satisfies the polynomial composition property. Let $p(z)$ be the polynomial defined in (\ref{eq:quotient-polynomial}), let $F = \sum_{i=0}^{\deg p -1} a_i \gamma_i$, and let $\varphi$ be the monoid isomorphism from $\mathcal{G}$ to $\mathcal{M}$ given in Theorem~\ref{thm:main}. In the following we construct several explicit shift-invariant permutations of $\mathbb F_2^n$ from binomial and trinomial permutations $\varphi(F)$, and then discuss  their orders, inverses, iterates, and fixed points.

\subsection{Shift-invariant permutations of $\mathbb F_2^n$ via permutation binomials in $\mathcal Q$}

In this subsection, we study shift-invariant permutations of the form $\beta_h=\operatorname{id}+\gamma_h$ with $h\ge1$. The images of such functions under the isomorphism $\varphi$ are permutation binomials $1+z^h \in \mathcal{Q}$.  By Lemma~\ref{lem:gamma-explicit}, for $h\geq 1$ we have \begin{equation}\label{eq:gamma-explicit-recall} \gamma_h=S^{hq}\bigodot_{i=0}^{h-1}S^{iq}\ell^{(q)}=S^{hq}\bigodot_{i=0}^{h-1}\left(\bigodot_{r\in R}(S^{iq+r}+\mathbf 1)\bigodot_{t\in T\setminus\{q\}}S^{iq+t}\right). \end{equation} Then, $\beta_h$ is a function from $\mathbb F_2^n$ to itself and its $i$-th output  coordinate is
\[\beta_h(x)_i=x_i+\gamma_h(x)_i=x_i+x_{i+hq}\prod_{j=0}^{h-1}\left(\prod_{r\in R}(x_{i+jq+r}+1)\prod_{t\in T\setminus\{q\}}x_{i+jq+t}\right).\]
Since $\varphi(\beta_h)=1+z^h$, by Theorem~\ref{thm:main} we know that $\beta_h$ is a permutation if and only if $\mathcal C\ne\emptyset$. In the following, we assume $\mathcal C\ne\emptyset$ and $h\in[1, j_0-1]$, where $j_0$ is given in Theorem~\ref{thm:gamma-behavior}.

For nonnegative integers $a$ and $b$ with binary expansions $a=\sum_{i\ge0}a_i2^i$ and $b=\sum_{i\ge0}b_i2^i$, write $a\preceq b$ if $a_i\le b_i$ for every $i$. The order and inverse formulas below follow from Theorem~\ref{thm:algebraic-properties} (3) and (1).  It is known that $(1+z^h)^v=\sum_{i=0}^{v}\binom{v}{i}z^{hi}$. By Theorem~\ref{thm:algebraic-properties}(2), we obtain $\beta_h^v=\sum_{i=0}^{v}\binom{v}{i}\gamma_{hi}$, where $\gamma_{hi}=\mathbf0$ for $hi\ge j_0$.  Lucas's theorem gives $\binom{v}{i}\equiv1\pmod2$ if and only if $i\preceq v$. For $j\ge0$ with $2^j h<j_0$, taking $v=2^j$ in the iterate formula gives $\beta_h^{2^j}=\operatorname{id}+\gamma_{2^j h}$. Hence $x$ is a fixed point of $\beta_h^{2^j}$ if and only if $\gamma_{2^j h}(x)=\mathbf0$.

\begin{proposition}\label{cor:bin}\label{th-fixed}
Let $\mathcal C\ne\emptyset$ and $h\in[1,j_0-1]$, and set $s=\lfloor(j_0-1)/h\rfloor$,  where $j_0$ is given in Theorem~\ref{thm:gamma-behavior}. Then
\[
	\operatorname{ord}(\beta_h)=2^{\lceil\log_2(j_0/h)\rceil},
	\quad
	\beta_h^{-1}=\operatorname{id}+\sum_{i=1}^{s}\gamma_{hi},
	\quad\text{and}\quad
	\beta_h^v=\sum_{\substack{i\in[0,\min\{v,s\}]\\i\preceq v}}\gamma_{hi},\, \,\, v\ge 1.
\]
In particular, $\beta_h$ is an involution if and only if $2h\ge j_0$. For $j\ge0$ with $2^j h<j_0$, $x\in\mathbb F_2^n$ is a fixed point of $\beta_h^{2^j}$ if and only if $\gamma_{2^j h}(x)=\mathbf0$.
\end{proposition}

We next apply Proposition~\ref{cor:bin} with $h=1$ to construct explicit permutations from three families of landscapes satisfying the quasi-conservation criterion in Lemma~\ref{lem:quasi-cons-set}.

\begin{corollary}\label{cor:bin-families}
Let $R,T\subseteq[-k,m]\setminus\{0\}$ be disjoint nonempty sets with $q\in T$, and let $n>\max\{2m+k,2k+m\}$, and $F:\mathbb F_2^n\to\mathbb F_2^n$ be the shift-invariant mapping whose $i$-th coordinate is
\[y_i=x_i+\prod_{r\in R}(x_{i+r}+1)\prod_{t\in T}x_{i+t}.\]
Suppose that one of the following three conditions holds:
\begin{enumerate}
\item[(1)] $R=\{r\}$ and $\{q,2r\}\subseteq T$, and for every $t\in T\setminus\{q\}$ we have $r+t\in T$ or $r-t\in T$;
\item[(2)] $T=\{q\}$ and for every $r\in R$ we have $q-r\in R$ or $q+r\in R$;
\item[(3)] $T=\{q,t\}$ with $q\ne t$. For every $r\in R$, at least one of $q-r$, $t-r$, $q+r$ and $t+r$ belongs to $R$, and at least one of $q-t$, $q+t$ and $2t$ belongs to $R$,
\end{enumerate}
where the addition is taken modulo $n$. Then $F$ is a permutation of $\mathbb F_2^n$ if and only if there exist $r\in R$ and $t\in T$ such that $r\equiv t\pmod{\gcd(n,q)}$.
When $F$ is a permutation, write $x=F^{-1}(y)$, and the $i$-th coordinate of its output is
\[x_i=y_i+\sum_{v=1}^{j_0-1}y_{i+vq}\prod_{j=0}^{v-1}\left(\prod_{r\in R}(y_{i+jq+r}+1)\prod_{t\in T\setminus\{q\}}y_{i+jq+t}\right),\]
where $j_0$ is defined in Theorem~\ref{thm:gamma-behavior}.
\end{corollary}

The three families of shift-invariant permutations above include several known constructions in~\cite{daemen1995cipher,haugland2026new,kriepke2026shift,liu2026finding,lyu2025generalized}. In particular, when $q=m$ and $R=[1,m-1]$,
the permutations from Corollary~\ref{cor:bin-families}(2) are exactly the $\chi_{n,m}$-functions introduced in \cite{lyu2025generalized}. 

The following example gives an explicit shift-invariant permutation from each of the three cases in Corollary~\ref{cor:bin-families}.

\begin{example}
The three shift-invariant permutations $F_j$ for $j\in[1,3]$ are defined by $y=F_j(x)$, where the $i$-th coordinate of the output $y$ is given as follows.
\begin{enumerate}
\item[(1)] For $F_1:\mathbb F_2^{10}\to\mathbb F_2^{10}$, take $R=\{1\}$, $T=\{-1,2\}$ and $q=2$, giving
\(y_i=x_i+x_{i+2}(x_{i+1}+1)x_{i-1}.\)
\item[(2)] For $F_2:\mathbb F_2^9\to\mathbb F_2^9$, take $R=\{-3,2\}$, $T=\{-1\}$ and $q=-1$, giving
\(y_i=x_i+x_{i+8}(x_{i+2}+1)(x_{i+6}+1).\)
\item[(3)] For $F_3:\mathbb F_2^{10}\to\mathbb F_2^{10}$, take $R=\{1\}$, $T=\{2,3\}$ and $q=3$, giving
\(y_i=x_i+x_{i+3}(x_{i+1}+1)x_{i+2}.\)
\end{enumerate}
Proposition~\ref{cor:bin} gives $\operatorname{ord}(F_1)=2$ and $\operatorname{ord}(F_2)=\operatorname{ord}(F_3)=4$. So $F_1$ is an involution, and its inverse  is itself.
	
For $x=F_2^{-1}(y)$, the $i$-th coordinate of the output is
\[x_i=y_i+(y_{i+2}+1)(y_{i+6}+1)\bigl(y_{i+8}+y_{i+7}(y_{i+1}+1)(y_{i+5}+1)\bigr).\]
For $x=F_3^{-1}(y)$, the $i$-th coordinate of the output is
\[x_i=y_i+(y_{i+1}+1)y_{i+2}\Bigl(y_{i+3}+(y_{i+4}+1)y_{i+5}\bigl(y_{i+6}+y_{i+9}(y_{i+7}+1)y_{i+8}\bigr)\Bigr).\]
Moreover, from Proposition~\ref{cor:bin} we have
$$ F_2^2(x)_i=x_i+x_{i+7}(x_{i+2}+1)(x_{i+6}+1)(x_{i+1}+1)(x_{i+5}+1), \,\,\, F_3^2(x)_i=x_i+x_{i+6}(x_{i+1}+1)x_{i+2}(x_{i+4}+1)x_{i+5}.$$
\end{example}

\subsection{Shift-invariant permutations of $\mathbb F_2^n$ via permutation trinomials in $\mathcal Q$}

In this subsection, we study shift-invariant permutations of $\mathbb F_2^n$ of the form $\kappa_u=\operatorname{id}+\gamma_u+\gamma_{2u}$ with $u\ge1$. By Theorem~\ref{thm:main}, $\varphi(\kappa_u)=1+z^u+z^{2u}$ and $\kappa_u$ is a permutation if and only if this trinomial is a unit in $\mathcal Q$. By \eqref{eq:gamma-explicit-recall}, the $i$-th output coordinate of $\kappa_u$ is
\begin{align*}
	\kappa_u(x)_i
	&=x_i+\gamma_u(x)_i+\gamma_{2u}(x)_i\\
	&=x_i+x_{i+uq}\prod_{j=0}^{u-1}
	\left(
	\prod_{r\in R}(x_{i+jq+r}+1)
	\prod_{t\in T\setminus\{q\}}x_{i+jq+t}
	\right)
	+x_{i+2uq}\prod_{j=0}^{2u-1}
	\left(
	\prod_{r\in R}(x_{i+jq+r}+1)
	\prod_{t\in T\setminus\{q\}}x_{i+jq+t}
	\right).
\end{align*}

In the following, we give the iterates, orders, and inverses of $\kappa_u$.

\begin{proposition}\label{prop:kappa-properties}
Let $j_0$ be as in Theorem~\ref{thm:gamma-behavior} and let $\kappa_u\in\mathcal G^*$ with $u\ge1$. If $\mathcal C\ne\emptyset$, we further assume $u<j_0$. Set $d=n/\gcd(n,q)$ and $d_u=d/\gcd(d,u)$. Then
\[\kappa_u^{\,v}=\sum_{a\preceq v}\sum_{b\preceq a}\gamma_{u(a+b)}, \,\,\, v \geq 1,
	\quad\text{and}\quad
	\operatorname{ord}(\kappa_u)=
	\begin{cases}
		2^{\lceil\log_2(j_0/u)\rceil},
		&\mathcal C\ne\emptyset,\\[6pt]
		d_u,
		&\mathcal C=\emptyset\text{ and }d\text{ is a power of two}.
	\end{cases}
	\]
	Moreover,
	\[
	\kappa_u^{-1}=
	\begin{cases}
		\displaystyle
		\operatorname{id}
		+\sum_{\substack{i\in[1,\lfloor(j_0-1)/u\rfloor]\\
				i\not\equiv2\pmod3}}\gamma_{iu},
		&\mathcal C\ne\emptyset,\\[10pt]
		\displaystyle
		\operatorname{id}
		+\sum_{\substack{i\in[1,d_u-1]\\
				i\not\equiv2\pmod3}}\gamma_{iu}
		+\sum_{\substack{i\in[d_u,2d_u-1]\\
				i+d_u\not\equiv2\pmod3}}\gamma_{iu},
		&\mathcal C=\emptyset.
	\end{cases}
	\]
\end{proposition}

\begin{IEEEproof}
By Theorem~\ref{thm:main}, we know that $\varphi(\kappa_u^v) = (1+z^u+z^{2u})^v $. By Lucas's theorem  we have
\[(1+z^u+z^{2u})^v=\bigl(1+z^u(1+z^u)\bigr)^v \equiv \sum_{a\preceq v}z^{ua}(1+z^u)^a \equiv \sum_{a\preceq v}\sum_{b\preceq a}z^{u(a+b)}\,\,\,  in\,\,\,  \mathbb F_2[z] . \]
Applying $\varphi^{-1}$ to the equality above, we obtain the iterate formula of $\kappa_u$.

If $\mathcal C\ne\emptyset$, the order formula of $\kappa_u$ follows from Theorem~\ref{thm:algebraic-properties} (3) with $\widehat j=u$. If $\mathcal C=\emptyset$ and $d$ is a power of two, the proof of Theorem~\ref{thm:algebraic-properties} (3) and $j_1\le d$ show that $\kappa_u^{2^\nu}=\operatorname{id}$ if and only if $d\mid2^\nu u$ for $\nu\ge0$. Since $d$ is a power of two, the smallest $2^\nu$ satisfying $d\mid2^\nu u$ is $d/\gcd(d,u)=d_u$. Hence $\operatorname{ord}(\kappa_u)=d_u$.
	
We now prove the inverse formulas.

\textbf{Case 1: $\mathcal C\ne\emptyset$.}
By direct calculations we have
\[ (1+z^u+z^{2u})(1+z^u) \sum_{i=0}^{\lfloor(j_0-1)/(3u)\rfloor}z^{3iu} =1+z^{3u(\lfloor(j_0-1)/(3u)\rfloor+1)} \equiv1\pmod{z^{j_0}}.\]
So, $ (1+z^u+z^{2u})^{-1} =(1+z^u)\sum_{i=0}^{\lfloor(j_0-1)/(3u)\rfloor}z^{3iu}$ in $\mathbb F_2[z] /(z^{j_0})$. Expanding this polynomial gives
\[ (1+z^u)\sum_{i=0}^{\lfloor(j_0-1)/(3u)\rfloor}z^{3iu} =\sum_{i=0}^{\lfloor(j_0-1)/(3u)\rfloor} \bigl(z^{3iu}+z^{(3i+1)u}\bigr) \equiv1+\sum_{\substack{i\in[1,\lfloor(j_0-1)/u\rfloor]\\ i\not\equiv2\pmod3}}z^{iu}\pmod{z^{j_0}}. \]
By Theorem~\ref{thm:main}, we have the inverse of $\kappa_u$:
$$ \kappa_u^{-1}= \operatorname{id}+\sum_{\substack{i\in[1,\lfloor(j_0-1)/u\rfloor]\\i\not\equiv2\pmod3}}\gamma_{iu}.$$

\textbf{Case 2: $\mathcal C=\emptyset$.}
We first show that $3\nmid d_u$. If $3\mid d_u$, then  $\gcd(d_u,u/\gcd(d,u))=1$ gives $3\nmid u/\gcd(d,u)$.
We have  $z^{3\gcd(d,u)}\equiv1\pmod{1+z^{\gcd(d,u)}+z^{2\gcd(d,u)}}$.  Since $d=d_u\gcd(d,u)$ and $3\mid d_u$, we obtain
\[ z^d=\left(z^{3\gcd(d,u)}\right)^{d_u/3} \equiv1\pmod{1+z^{\gcd(d,u)}+z^{2\gcd(d,u)}}. \]
If $u/\gcd(d,u)\equiv1\pmod3$, writing  $u=(3a+1)\gcd(d,u)$ for some integer $a\ge0$, then
\[ z^u=\left(z^{3\gcd(d,u)}\right)^a z^{\gcd(d,u)} \equiv z^{\gcd(d,u)} \pmod{1+z^{\gcd(d,u)}+z^{2\gcd(d,u)}} \]
and $z^{2u}\equiv z^{2\gcd(d,u)} \pmod{1+z^{\gcd(d,u)}+z^{2\gcd(d,u)}}$. Hence,
\[ 1+z^u+z^{2u} \equiv1+z^{\gcd(d,u)}+z^{2\gcd(d,u)} \equiv0\pmod{1+z^{\gcd(d,u)}+z^{2\gcd(d,u)}}. \]
Similarly, if $u/\gcd(d,u)\equiv2\pmod3$, then $1+z^u+z^{2u}\equiv 0 \pmod{1+z^{\gcd(d,u)}+z^{2\gcd(d,u)}}$. Thus $1+z^{\gcd(d,u)}+z^{2\gcd(d,u)}$ divides both $1+z^u+z^{2u}$ and $p(z)=z^{j_1}(1+z^d)$.
So $1+z^u+z^{2u}$ is not a unit in $\mathcal Q$. Hence $\kappa_u$ is not a permutation by Theorem~\ref{thm:main}. This contradicts the assumption that $\kappa_u\in\mathcal G^*$, and so $3\nmid d_u$.

To verify the inverse formula, let
\[ H(z)= \sum_{\substack{i\in[0,d_u-1]\\i\not\equiv2\pmod3}}z^{iu} +\sum_{\substack{i\in[d_u,2d_u-1]\\i+d_u\not\equiv2\pmod3}}z^{iu}. \]
For $d_u=1$ and $2$, we have $H(z)=1$ and $H(z)=1+z^u+z^{2u}$, respectively. So $(1+z^u+z^{2u})H(z)=1+z^{d_u u}+z^{2d_u u}$.

For $d_u>2$ with $d_u\equiv 1 \pmod3$, we have
$$H(z)=(1+z^u)\sum_{i=0}^{(d_u-4)/3}z^{3iu}+z^{(d_u-1)u}+z^{(d_u+1)u}(1+z^u)\sum_{i=0}^{(d_u-4)/3}z^{3iu},$$
and so,
\[(1+z^u+z^{2u})H(z)=(1+z^{3u})(1+z^{(d_u+1)u})\sum_{i=0}^{(d_u-4)/3}z^{3iu}+z^{(d_u-1)u}(1+z^u+z^{2u})=1+z^{d_u u}+z^{2d_u u}.\]
Similarly, for $d_u>2$ with $d_u\equiv 2 \pmod3$, we have
$$H(z)=(1+z^u)\sum_{i=0}^{(d_u-2)/3}z^{3iu}+z^{d_u u}+z^{(d_u+2)u}(1+z^u)\sum_{i=0}^{(d_u-5)/3}z^{3iu}$$
and
\[(1+z^u+z^{2u})H(z)=1+z^{d_u u}+z^{2d_u u}.\]
Thus, for every $d_u\ge1$ with $3\nmid d_u$, we have
\[
(1+z^u+z^{2u})H(z)
=1+z^{d_u u}(1+z^{d_u u})
\equiv1\pmod{z^{j_1}(1+z^d)}.
\]
The congruence holds because $d\mid d_u u$ and
$d_u u\ge d\ge j_1$. Thus $H(z)$ represents the inverse of
$1+z^u+z^{2u}$ in $\mathcal Q$.
By Theorem~\ref{thm:main}, the inverse of $\kappa_u$ is
\[
\kappa_u^{-1}
=\operatorname{id}
+\sum_{\substack{i\in[1,d_u-1]\\i\not\equiv2\pmod3}}\gamma_{iu}
+\sum_{\substack{i\in[d_u,2d_u-1]\\i+d_u\not\equiv2\pmod3}}\gamma_{iu}.
\]

\end{IEEEproof}

We next apply Proposition~\ref{prop:kappa-properties} with $u=1$ to construct explicit permutations from three families of landscapes satisfying the quasi-conservation criterion in Lemma~\ref{lem:quasi-cons-set}. The irreducible polynomial $1+z+z^2$ has roots of order $3$. Theorem~\ref{thm:main} and Corollary~\ref{cor:pp-period} show that $F=\operatorname{id}+\gamma_1+\gamma_2$ is a permutation if and only if $\mathcal C\ne\emptyset$ or $3\nmid d$.

\begin{corollary}\label{cor:trinomial-families}
Let $R,T\subseteq[-k,m]\setminus\{0\}$ be disjoint nonempty sets with $q\in T$, and let $n>\max\{2m+k,2k+m\}$, and $F:\mathbb F_2^n\to\mathbb F_2^n$ be the shift-invariant mapping whose $i$-th coordinate is
\[y_i=x_i+\prod_{r\in R}(x_{i+r}+1)\prod_{t\in T}x_{i+t}+x_{i+2q}\prod_{r\in R}(x_{i+r}+1)(x_{i+q+r}+1)\prod_{t\in T\setminus\{q\}}x_{i+t}x_{i+q+t}.\]
Suppose that one of the following three conditions holds:
\begin{enumerate}
\item[(1)] $R=\{r\}$ and $\{q,2r\}\subseteq T$, and for every $t\in T\setminus\{q\}$ we have $r+t\in T$ or $r-t\in T$;
\item[(2)] $T=\{q\}$ and for every $r\in R$ we have $q-r\in R$ or $q+r\in R$;
\item[(3)] $T=\{q,t\}$ with $q\ne t$. For every $r\in R$, at least one of $q-r$, $t-r$, $q+r$ and $t+r$ belongs to $R$, and at least one of $q-t$, $q+t$ and $2t$ belongs to $R$,
\end{enumerate}
where the addition is taken modulo $n$. Then $F$ is a permutation of $\mathbb F_2^n$ if and only if either $r\equiv t\pmod{\gcd(n,q)}$ for some $r\in R$ and $t\in T$ or $3\nmid d$, where $d=n/\gcd(n,q)$.
\end{corollary}

The following example gives explicit shift-invariant permutations from each of the three cases in Corollary~\ref{cor:trinomial-families}.
\begin{example}
The three shift-invariant permutations $F_j$ for $j\in[4,6]$ are defined by $y=F_j(x)$, where the $i$-th coordinate of the output $y$ is given as follows.
	\begin{enumerate}
		\item[(1)] For $F_4:\mathbb F_2^{12}\to\mathbb F_2^{12}$, take $R=\{-2\}$, $T=\{-4,2,3\}$ and $q=3$, giving
		\[
		y_i=x_i+x_{i+3}(x_{i-2}+1)x_{i-4}x_{i+2}
		+x_{i+6}(x_{i-2}+1)x_{i-4}x_{i+2}
		(x_{i+1}+1)x_{i-1}x_{i+5}.
		\]
		
		\item[(2)] For $F_5:\mathbb F_2^{10}\to\mathbb F_2^{10}$, take $R=\{1,3\}$, $T=\{4\}$ and $q=4$, giving
		\[
		y_i=x_i+x_{i+4}(x_{i+1}+1)(x_{i+3}+1)
		+x_{i+8}\prod_{j=0}^{3}(x_{i+2j+1}+1).
		\]
		
		\item[(3)] For $F_6:\mathbb F_2^8\to\mathbb F_2^8$, take $R=\{-1,2\}$, $T=\{1,3\}$ and $q=3$, giving
		\[
		y_i=x_i+x_{i+3}(x_{i-1}+1)(x_{i+2}+1)x_{i+1}
		+x_{i+6}(x_{i-1}+1)(x_{i+2}+1)(x_{i+5}+1)x_{i+1}x_{i+4}.
		\]
	\end{enumerate}
		For $F_4$ and $F_5$, we have $\mathcal C=\emptyset$, with $(j_1,d)=(4,4)$ and $(3,5)$, respectively. For $F_6$, we have $\mathcal C\ne\emptyset$ and $j_0=3$. Proposition~\ref{prop:kappa-properties} gives $\operatorname{ord}(F_4)=\operatorname{ord}(F_6)=4$. The polynomial $1+z+z^2$ has order $4$ modulo $z^3$ and order $15$ modulo $1+z^5$. By the Chinese Remainder Theorem and Theorem~\ref{thm:main}, we have $\operatorname{ord}(F_5)=\operatorname{lcm}(4,15)=60$.
	
Proposition~\ref{prop:kappa-properties} gives
	\(
	F_4^{-1}=\operatorname{id}+\gamma_1+\gamma_3+\gamma_5+\gamma_6\),
	\(F_5^{-1}=\operatorname{id}+\gamma_1+\gamma_4+\gamma_5+\gamma_7\) and
	\(F_6^{-1}=\operatorname{id}+\gamma_1.
	\)
	For $x=F_4^{-1}(y)$, the $i$-th coordinate of the output is
	\[
	x_i=y_i+\sum_{v\in\{1,3,5,6\}}y_{i+3v}
	\prod_{j=0}^{v-1}
	\left((y_{i+3j-2}+1)y_{i+3j-4}y_{i+3j+2}\right).
	\]
	For $x=F_5^{-1}(y)$, the $i$-th coordinate of the output is
	\[
	x_i=y_i+y_{i+4}(y_{i+1}+1)(y_{i+3}+1)
	+(y_i+y_{i+6}+y_{i+8})\prod_{j=0}^{4}(y_{i+2j+1}+1).
	\]
	For $x=F_6^{-1}(y)$, the $i$-th coordinate of the output is
	\[
	x_i=y_i+y_{i+3}(y_{i-1}+1)(y_{i+2}+1)y_{i+1}.
	\]
Moreover, from Proposition~\ref{prop:kappa-properties} we have
	\begin{align*}
		F_4^2(x)_i
		&=x_i+\sum_{v\in\{2,4\}}x_{i+3v}
		\prod_{j=0}^{v-1}
		\left((x_{i+3j-2}+1)x_{i+3j-4}x_{i+3j+2}\right),\\
		F_5^2(x)_i
		&=x_i+x_{i+8}\prod_{j=0}^{3}(x_{i+2j+1}+1)
		+x_{i+6}\prod_{j=0}^{4}(x_{i+2j+1}+1),\\
		F_6^2(x)_i
		&=x_i+x_{i+6}(x_{i-1}+1)(x_{i+2}+1)(x_{i+5}+1)x_{i+1}x_{i+4}.
	\end{align*}
\end{example}

\subsection{Representative families of permutations on $\mathbb F_2^n$}
\label{subsec:tables}

In this subsection, we give representative families of shift-invariant permutations of the forms $\beta_1=\operatorname{id}+\gamma_1$ and $\kappa_1=\operatorname{id}+\gamma_1+\gamma_2$ from the preceding two subsections. Recall that \eqref{eq:gamma-explicit-recall} gives
\[
\gamma_1(x)_i=\prod_{r\in R}(x_{i+r}+1)\prod_{t\in T}x_{i+t},
\quad
\gamma_2(x)_i=x_{i+2q}
\prod_{r\in R}(x_{i+r}+1)(x_{i+q+r}+1)
\prod_{t\in T\setminus\{q\}}x_{i+t}x_{i+q+t}.
\]
To determine the dimensions of vector spaces in which $\beta_1$ and $\kappa_1$ are permutations, we express the criterion in Theorem~\ref{thm:main} as conditions on $n$ in Proposition~\ref{prop:xi-tables}. We list these families in Tables~\ref{tab:landscapes}, \ref{tab:landscapes-Tq}, and~\ref{tab:landscapes-Tqt}. All listed landscapes satisfy $|\supp|\le4$ and $\supp\subseteq[-4,4]$.

We use the following equivalences to reduce repetition in the tables. Elementary equivalent functions~\cite[Definition~2.10]{haugland2025shift} can be defined by two landscapes related by reversing the string or exchanging $0$ and $1$. These functions are permutations in the same dimensions. 
Multiplying all offsets by an integer $a$ with $\gcd(a,n)=1$ also gives an equivalent mapping under the change of coordinates $x_i\mapsto x_{ai}$. For example, the landscapes $1{*}01$ and $1{-}{*}{-}0{-}1$ are equivalent with $a=2$ when $n$ is odd, and the tables include only $1{*}01$.

For fixed $R,T$ and $q$, let $F$ denote the family of mappings defined by $\beta_1$ or $\kappa_1$ for $n\ge1$. Suppose that two distinct inputs in $\mathbb F_2^t$ have the same image. Repeating each input $v$ times gives two distinct inputs in $\mathbb F_2^{vt}$ with the same image for every $v\ge1$. Thus the dimensions of vector spaces in which $F$ is not a permutation are closed under taking multiples. Let $\xi(F)$ consist of the divisibility-minimal positive integers $t$ for which $F$ is not a permutation of $\mathbb F_2^t$, as in~\cite{daemen1995cipher,kriepke2025siblings}. Then $F$ is a permutation of $\mathbb F_2^n$ if and only if $t\nmid n$ for every $t\in\xi(F)$. For example, $\xi(F)=\{2\}$ means that $F$ is a permutation exactly when $n$ is odd. The following proposition determines $\xi(\beta_1)$ and $\xi(\kappa_1)$ using sets that do not depend on $n$.

\begin{proposition}\label{prop:xi-tables}
Suppose that $\ell$ is quasi-conserved with special index $q$ on $\mathbb F_2^n$ for some $n>\max\{2m+k,2k+m\}$. Set
\[\mathcal D=\left\{g\ge1:g\mid q,\quad r\not\equiv t\pmod g \text{ for every }r\in R\text{ and }t\in T \right\}. \]
For a nonzero integer $a$, let $v_3(a)$ be the largest integer $e$ such that $3^e\mid a$. Then
\begin{enumerate}
\item[(1)] $\xi(\beta_1)$ consists of the divisibility-minimal elements of $\mathcal D$, where an element is divisibility-minimal if no smaller element of the same set divides it.
\item[(2)] $\xi(\kappa_1)$ consists of the divisibility-minimal elements of
\[\left\{3^{v_3(q/g)+1}g:g\in\mathcal D\right\}.\]
\end{enumerate}
\end{proposition}

\begin{IEEEproof}
The proofs for $\beta_1$ and $\kappa_1$ are similar, so we give only the proof of (2). For \(n>\max\{2m+k,2k+m\}\), we first establish a divisibility criterion that guarantees
\(\kappa_1\) is not a permutation. Let $q_0=\gcd(n,q)$ and $d=n/q_0$. Recall that $\mathcal C=\{(r,t)\in R\times T:r\equiv t\pmod{q_0}\}$. Hence $\mathcal C=\emptyset$ if and only if $r\not\equiv t\pmod{q_0}$ for every $(r,t)\in R\times T$, which holds if and only if $q_0\in\mathcal D$.
	
	If $\mathcal C=\emptyset$, then $p(z)=z^{j_1}(1+z^d)$. The irreducible polynomial $1+z+z^2$ has roots of order $3$, so it has a common factor with $p(z)$ if and only if $3\mid d$. By Theorem~\ref{thm:main} and Corollary~\ref{cor:pp-period}, $\kappa_1$ is not a permutation if and only if $q_0\in\mathcal D$ and $3\mid d$.
	Since $v_3(d)=v_3(n)-v_3(q_0)=v_3(n)-\min\{v_3(n),v_3(q)\}$, we have $3\mid d$ if and only if $v_3(n)>v_3(q)$. For each $g\in\mathcal D$, the condition $3^{v_3(q/g)+1}g\mid n$ is equivalent to $g\mid n$ and $v_3(n)-v_3(g)\ge v_3(q)-v_3(g)+1$, or equivalently, $g\mid n$ and $v_3(n)>v_3(q)$.
	It is known that $q_0\in\mathcal D$ if and only if $g\mid n$ for some $g\in\mathcal D$, and $3\mid d$ if and only if $v_3(n)>v_3(q)$. This shows that there exists $g\in\mathcal D$ satisfying $3^{v_3(q/g)+1}g\mid n$ if and only if $q_0\in\mathcal D$ and $3\mid d$. This is exactly the condition for $\kappa_1$ not being a permutation 
for the case $\mathcal C=\emptyset$.
		
If $\mathcal C\ne\emptyset$, then $\kappa_1$ is a permutation by Theorem~\ref{thm:main}. Since $q_0\notin\mathcal D$, we have $g\nmid n$ and hence $3^{v_3(q/g)+1}g\nmid n$ for every $g\in\mathcal D$. This proves the divisibility criterion for $\kappa_1$.
We now extend this divisibility criterion to all dimensions of vector spaces. By Lemma~\ref{lem:quasi-cons-set} and $n>\max\{2m+k,2k+m\}$, for every $\delta\in\supp\setminus\{q\}$ there exist $r\in R$ and $t\in T$ such that $r+\delta=t$ or $t+\delta=r$. For any $t\ge1$, choose $v\ge1$ such that $tv>\max\{2m+k,2k+m\}$. By Theorem~\ref{thm:equivalence}, the polynomial composition property holds on $\mathbb F_2^{tv}$. Restricting the maps to inputs consisting of $v$ copies of a vector in $\mathbb F_2^t$ gives the polynomial composition property on $\mathbb F_2^t$.
		
For each $g\in\mathcal D$, we next show that $\kappa_1$ is not a permutation of $\mathbb F_2^h$ with $h=3^{v_3(q/g)+1}g$. Since $\gcd(h,q)=h/3$, we have $S^{3q}=\operatorname{id}$ on $\mathbb F_2^h$. We next show that $\gamma_3+\gamma_4\ne\mathbf 0$ on $\mathbb F_2^h$. The quasi-conservation condition implies 
$g\nmid\delta$ for every $\delta\in\supp\setminus\{q\}$. Otherwise, an equality $r+\delta=t$ or $t+\delta=r$ would give $r\equiv t\pmod g$. This is contrary to $g\in\mathcal D$.
By~\eqref{eq:gamma-explicit-recall} and $S^{3q}=\operatorname{id}$, we have
		\begin{align*}
				\gamma_3+\gamma_4
				&=\operatorname{id}\bigodot_{j=0}^{2}S^{jq}\ell^{(q)}
				+S^q\odot\ell^{(q)}\bigodot_{j=0}^{2}S^{jq}\ell^{(q)}
				=\operatorname{id}\bigodot_{j=0}^{2}S^{jq}\ell^{(q)}
				+S^q\bigodot_{j=0}^{2}S^{jq}\ell^{(q)}\\
				&=(\operatorname{id}+S^q)\bigodot_{j=0}^{2}
				\left(
				\bigodot_{r\in R}(S^{jq+r}+\mathbf1)
				\bigodot_{t\in T\setminus\{q\}}S^{jq+t}
				\right).
		\end{align*}
Since $g\mid h$ and $g\mid q$, the congruence $r+jq\equiv t+kq\pmod h$ implies that $r\equiv t\pmod g$. Hence $\bigodot_{j=0}^{2}S^{jq}\ell^{(q)}\ne\mathbf0$. Since $\gcd(h,q)=h/3$, we have $\operatorname{id}+S^q\ne\mathbf0$. The Hadamard product of two nonzero mappings may be zero, so we now rule out this possibility. Since $g\nmid\delta$ for every $\delta\in\supp\setminus\{q\}$, no factor in $\bigodot_{j=0}^{2}S^{jq}\ell^{(q)}$ is $\operatorname{id}$, $S^q$, $\operatorname{id}+\mathbf1$ or $S^q+\mathbf1$. Therefore, $\gamma_3+\gamma_4\ne\mathbf0$. By \eqref{eq:gamma-explicit-recall},
			\[
			\gamma_6
			=S^{6q}\bigodot_{j=0}^{5}S^{jq}\ell^{(q)}
			=S^{6q}\bigodot_{j=0}^{2}
			\bigl(S^{jq}\ell^{(q)}\odot S^{(j+3)q}\ell^{(q)}\bigr)
			=S^{3q}\bigodot_{j=0}^{2}
			\bigl(S^{jq}\ell^{(q)}\odot S^{jq}\ell^{(q)}\bigr)
			=S^{3q}\bigodot_{j=0}^{2}S^{jq}\ell^{(q)}
			=\gamma_3.
			\]
The polynomial composition property gives
			\[
			(\gamma_3+\gamma_4)\circ\kappa_1
			=(\gamma_3+\gamma_4)\circ(\operatorname{id}+\gamma_1+\gamma_2)
			=\gamma_3+\gamma_6
			=\mathbf0.
			\]
So, $\kappa_1$ is not a permutation of $\mathbb F_2^h$.
		
We have proved that $\kappa_1$ is not a permutation of vector spaces over $\F_2$ with dimensions divisible by an element of $\{3^{v_3(q/g)+1}g:g\in\mathcal D\}$. We now prove that $\kappa_1$ is a permutation of $\mathbb F_2^t$ when $t$ is not divisible by any element of \(\left\{3^{v_3(q/g)+1}g:g\in\mathcal D\right\}\). Choose $v$ such that $\gcd(v,3|q|)=1$ and $tv>\max\{2m+k,2k+m\}$. For every $g\in\mathcal D$, we can write $|q|=3^{v_3(q/g)} gu$ for some positive integer $u$, so $3^{v_3(q/g)+1}g$ divides $3|q|$. Hence, $3^{v_3(q/g)+1}g$ does not divide $tv$ for any $g\in\mathcal D$, since it is coprime to $v$ and does not divide $t$. By the conclusions proved above, $\kappa_1$ is a permutation of $\mathbb F_2^{tv}$. 
Since the dimensions of vector spaces in which $\kappa_1$ is not a permutation are closed under taking multiples, $\kappa_1$ is also a permutation of $\mathbb F_2^t$; otherwise, it would not be a permutation of $\mathbb F_2^{tv}$.
Thus, by the definition, $\xi(\kappa_1)$ consists of the divisibility-minimal elements of \(\left\{3^{v_3(q/g)+1}g:g\in\mathcal D\right\}\). This proves (2).
\end{IEEEproof}

We now explain Tables~\ref{tab:landscapes}, \ref{tab:landscapes-Tq}, and~\ref{tab:landscapes-Tqt}. $R$ and $T$ are respectively the sets of offsets corresponding to $0$ and $1$ in the landscape notation of Section~\ref{sec:prelim}, and $\supp=R\cup T$ is the support of the landscape. The parameter $q\in T$ is a special index. For a conserved landscape, the column $q$ contains ``--'' and any $q\in T$ may be chosen. The sets $\xi(\beta_1)$ and $\xi(\kappa_1)$ list the divisibility-minimal dimensions of the vector spaces on which $\beta_1$ and $\kappa_1$ are not permutations, respectively. Taking the last row in Table~\ref{tab:landscapes} as an example, we use an example to illustrate the meaning of the three tables.
\begin{example}
The last row of Table~\ref{tab:landscapes} gives the quasi-conserved landscape $1{*}{-}011$ with $R=\{2\}$, $T=\{-1,3,4\}$ and $q=4$. $\xi(\beta_1)=\{4\}$ means that $\beta_1$ is a permutation if and only if $4\nmid n$, while $\xi(\kappa_1)=\{12\}$ means that $\kappa_1$ is a permutation if and only if $12\nmid n$. The corresponding mappings from $\mathbb F_2^n$ to itself have the following representations:
\[\beta_1=\operatorname{id}+(S^2+\mathbf1)\odot S^{-1}\odot S^3\odot S^4,\quad \kappa_1=\operatorname{id}+(S^2+\mathbf1)\odot S^{-1}\odot S^3\odot \bigl(S^4+(S^6+\mathbf1)\odot S^7\odot S^8\bigr). \]
Their $i$-th coordinate functions are given by
\[\beta_1(x)_i=x_i+(x_{i+2}+1)x_{i-1}x_{i+3}x_{i+4},\quad \kappa_1(x)_i=x_i+(x_{i+2}+1)x_{i-1}x_{i+3} \bigl(x_{i+4}+(x_{i+6}+1)x_{i+7}x_{i+8}\bigr),\]
where all coordinate subscripts are taken modulo $n$.
\end{example}

\begin{table}[htbp]
	\centering
	\caption{SOME PERMUTATIONS FROM COROLLARIES~\ref{cor:bin-families}(1) AND~\ref{cor:trinomial-families}(1) WITH $|R|=1$}
	\label{tab:landscapes}
	\begin{tabular}{ccccc|c c|c c}
		\toprule
		$R$ & $T$ & $q$ & property & Landscape &
		$\xi(\beta_1)$ & Ref. & $\xi(\kappa_1)$ & Ref. \\
		\midrule
		$\{1\}$ & $\{2\}$ & $2$ & quasi-conserved & $*01$ & $\{2\}$ & \cite{daemen1995cipher,kriepke2024algebraic,kriepke2026shift,liu2022inverse,liu2026finding,lyu2025generalized} & $\{6\}$ & \cite{daemen1995cipher,kriepke2025siblings} \\
		\midrule
		$\{1\}$ & $\{-1,2\}$ & -- & conserved & $1*01$ & $\emptyset$ & \cite{daemen1995cipher,kriepke2026shift} & $\emptyset$ & -- \\
		$\{1\}$ & $\{2,3\}$ & $3$ & quasi-conserved & $*011$ & $\{3\}$ & \cite{daemen1995cipher,kriepke2026shift,liu2026finding} & $\{9\}$ & -- \\
		$\{2\}$ & $\{1,4\}$ & $4$ & quasi-conserved & $*10-1$ & $\{4\}$ & \cite{daemen1995cipher,kriepke2026shift} & $\{12\}$ & -- \\
		\midrule
		$\{1\}$ & $\{-4,-1,2\}$ & $-4$ & quasi-conserved & $1--1*01$ & $\{4\}$ & -- & $\{12\}$ & -- \\
		$\{1\}$ & $\{-3,-1,2\}$ & $-3$ & quasi-conserved & $1-1*01$ & $\{3\}$ & -- & $\{9\}$ & -- \\
		$\{1\}$ & $\{-3,2,4\}$ & $2$ & quasi-conserved & $1--*01-1$ & $\emptyset$ & -- & $\emptyset$ & -- \\
		$\{1\}$ & $\{-2,-1,2\}$ & -- & conserved & $11*01$ & $\emptyset$ & \cite{daemen1995cipher,liu2026finding} & $\emptyset$ & -- \\
		$\{1\}$ & $\{-2,2,3\}$ & -- & conserved & $1-*011$ & $\emptyset$ & \cite{daemen1995cipher,haugland2026new,kriepke2026shift,liu2026finding} & $\emptyset$ & -- \\
		$\{1\}$ & $\{-1,2,3\}$ & $3$ & quasi-conserved & $1*011$ & $\{3\}$ & \cite{daemen1995cipher,kriepke2026shift} & $\{9\}$ & -- \\
		$\{1\}$ & $\{-1,2,4\}$ & $4$ & quasi-conserved & $1*01-1$ & $\{4\}$ & -- & $\{12\}$ & -- \\
		$\{1\}$ & $\{2,3,4\}$ & $4$ & quasi-conserved & $*0111$ & $\{4\}$ & \cite{daemen1995cipher,liu2026finding} & $\{12\}$ & -- \\
		$\{2\}$ & $\{-3,-2,4\}$ & $-3$ & quasi-conserved & $11-*-0-1$ & $\{3\}$ & -- & $\{9\}$ & -- \\
		$\{2\}$ & $\{-2,-1,4\}$ & $-1$ & quasi-conserved & $11*-0-1$ & $\emptyset$ & -- & $\emptyset$ & -- \\
		$\{2\}$ & $\{-2,1,4\}$ & -- & conserved & $1-*10-1$ & $\emptyset$ & -- & $\emptyset$ & -- \\
		$\{2\}$ & $\{-2,3,4\}$ & $3$ & quasi-conserved & $1-*-011$ & $\{3\}$ & -- & $\{9\}$ & -- \\
		$\{2\}$ & $\{-1,1,4\}$ & $4$ & quasi-conserved & $1*10-1$ & $\{4\}$ & \cite{daemen1995cipher,kriepke2026shift,liu2026finding} & $\{12\}$ & -- \\
		$\{2\}$ & $\{-1,3,4\}$ & $4$ & quasi-conserved & $1*-011$ & $\{4\}$ & -- & $\{12\}$ & -- \\
		\bottomrule
	\end{tabular}
\end{table}

\begin{table}[htbp]
	\centering
	\caption{SOME PERMUTATIONS FROM COROLLARIES~\ref{cor:bin-families}(2) AND~\ref{cor:trinomial-families}(2) WITH $|T|=1$}
	\label{tab:landscapes-Tq}
	\begin{tabular}{ccccc|c c|c c}
		\toprule
		$R$ & $T$ & $q$ & property & Landscape &
		$\xi(\beta_1)$ & Ref. & $\xi(\kappa_1)$ & Ref. \\
		\midrule
		$\{-3,4\}$ & $\{1\}$ & $1$ & quasi-conserved & $0--*1--0$ & $\emptyset$ & -- & $\emptyset$ & -- \\
		$\{-2,3\}$ & $\{1\}$ & $1$ & quasi-conserved & $0-*1-0$ & $\emptyset$ & -- & $\emptyset$ & -- \\
		$\{-1,1\}$ & $\{2\}$ & $2$ & quasi-conserved & $0*01$ & $\{2\}$ & \cite{daemen1995cipher,kriepke2026shift,liu2026finding} & $\{6\}$ & -- \\
		$\{-1,3\}$ & $\{2\}$ & $2$ & quasi-conserved & $0*-10$ & $\{2\}$ & \cite{daemen1995cipher,kriepke2026shift,liu2026finding} & $\{6\}$ & -- \\
		$\{-1,4\}$ & $\{3\}$ & $3$ & quasi-conserved & $0*--10$ & $\{3\}$ & -- & $\{9\}$ & -- \\
		$\{1,2\}$ & $\{3\}$ & $3$ & quasi-conserved & $*001$ & $\{3\}$ & \cite{daemen1995cipher,kriepke2026shift,liu2026finding,lyu2025generalized} & $\{9\}$ & -- \\
		$\{1,3\}$ & $\{4\}$ & $4$ & quasi-conserved & $*0-01$ & $\{2\}$ & \cite{daemen1995cipher,kriepke2026shift,liu2026finding} & $\{6\}$ & -- \\
		\midrule
		$\{-4,-3,4\}$ & $\{1\}$ & $1$ & quasi-conserved & $00--*1--0$ & $\emptyset$ & -- & $\emptyset$ & -- \\
		$\{-3,-2,3\}$ & $\{1\}$ & $1$ & quasi-conserved & $00-*1-0$ & $\emptyset$ & -- & $\emptyset$ & -- \\
		$\{-3,3,4\}$ & $\{1\}$ & $1$ & quasi-conserved & $0--*1-00$ & $\emptyset$ & -- & $\emptyset$ & -- \\
		$\{-3,-1,1\}$ & $\{2\}$ & $2$ & quasi-conserved & $0-0*01$ & $\{2\}$ & -- & $\{6\}$ & -- \\
		$\{-3,-1,3\}$ & $\{2\}$ & $2$ & quasi-conserved & $0-0*-10$ & $\{2\}$ & -- & $\{6\}$ & -- \\
		$\{-1,1,3\}$ & $\{2\}$ & $2$ & quasi-conserved & $0*010$ & $\{2\}$ & \cite{daemen1995cipher,kriepke2026shift,liu2026finding} & $\{6\}$ & -- \\
		$\{-4,-1,4\}$ & $\{3\}$ & $3$ & quasi-conserved & $0--0*--10$ & $\{3\}$ & -- & $\{9\}$ & -- \\
		$\{-2,1,2\}$ & $\{3\}$ & $3$ & quasi-conserved & $0-*001$ & $\{3\}$ & -- & $\{9\}$ & -- \\
		$\{-1,1,4\}$ & $\{3\}$ & $3$ & quasi-conserved & $0*0-10$ & $\{3\}$ & -- & $\{9\}$ & -- \\
		$\{-1,1,2\}$ & $\{3\}$ & $3$ & quasi-conserved & $0*001$ & $\{3\}$ & \cite{daemen1995cipher,kriepke2026shift,liu2026finding} & $\{9\}$ & -- \\
		$\{-3,1,3\}$ & $\{4\}$ & $4$ & quasi-conserved & $0--*0-01$ & $\{2\}$ & -- & $\{6\}$ & -- \\
		$\{-1,1,3\}$ & $\{4\}$ & $4$ & quasi-conserved & $0*0-01$ & $\{2\}$ & -- & $\{6\}$ & -- \\
		$\{1,2,3\}$ & $\{4\}$ & $4$ & quasi-conserved & $*0001$ & $\{4\}$ & \cite{daemen1995cipher,kriepke2026shift,liu2026finding,lyu2025generalized} & $\{12\}$ & -- \\
		\bottomrule
	\end{tabular}
\end{table}

\begin{table}[htbp]
	\centering
	\caption{SOME PERMUTATIONS FROM COROLLARIES~\ref{cor:bin-families}(3) AND~\ref{cor:trinomial-families}(3) WITH $|T|=2$}
	\label{tab:landscapes-Tqt}
	\begin{tabular}{ccccc|c c|c c}
		\toprule
		$R$ & $T$ & $q$ & property & Landscape &
		$\xi(\beta_1)$ & Ref. & $\xi(\kappa_1)$ & Ref. \\
		\midrule
		$\{-3,2\}$ & $\{1,-1\}$ & -- & conserved & $0-1*10$ & $\emptyset$ & \cite{haugland2026new} & $\emptyset$ & -- \\
		$\{2,3\}$ & $\{1,-1\}$ & -- & conserved & $1*100$ & $\emptyset$ & \cite{daemen1995cipher,kriepke2026shift,liu2026finding} & $\emptyset$ & -- \\
		$\{-3,4\}$ & $\{1,2\}$ & $1$ & quasi-conserved & $0--*11-0$ & $\emptyset$ & -- & $\emptyset$ & -- \\
		$\{-2,4\}$ & $\{1,2\}$ & $1$ & quasi-conserved & $0-*11-0$ & $\emptyset$ & -- & $\emptyset$ & -- \\
		$\{-2,3\}$ & $\{1,2\}$ & -- & conserved & $0-*110$ & $\emptyset$ & \cite{haugland2026new} & $\emptyset$ & -- \\
		$\{-1,3\}$ & $\{1,2\}$ & -- & conserved & $0*110$ & $\emptyset$ & \cite{daemen1995cipher,kriepke2026shift} & $\emptyset$ & -- \\
		$\{-4,2\}$ & $\{1,-2\}$ & -- & conserved & $0-1-*10$ & $\emptyset$ & -- & $\emptyset$ & -- \\
		$\{-1,2\}$ & $\{1,-2\}$ & -- & conserved & $10*10$ & $\emptyset$ & \cite{daemen1995cipher,kriepke2026shift,liu2026finding} & $\emptyset$ & -- \\
		$\{-3,4\}$ & $\{1,3\}$ & -- & conserved & $0--*1-10$ & $\emptyset$ & -- & $\emptyset$ & -- \\
		$\{-1,4\}$ & $\{1,3\}$ & -- & conserved & $0*1-10$ & $\emptyset$ & \cite{haugland2026new} & $\emptyset$ & -- \\
		$\{-2,3\}$ & $\{1,-3\}$ & -- & conserved & $10-*1-0$ & $\emptyset$ & -- & $\emptyset$ & -- \\
		$\{-3,4\}$ & $\{1,-4\}$ & -- & conserved & $10--*1--0$ & $\emptyset$ & -- & $\emptyset$ & -- \\
		$\{-4,3\}$ & $\{2,-1\}$ & $2$ & quasi-conserved & $0--1*-10$ & $\emptyset$ & -- & $\emptyset$ & -- \\
		$\{-1,4\}$ & $\{2,3\}$ & -- & conserved & $0*-110$ & $\emptyset$ & \cite{haugland2026new} & $\emptyset$ & -- \\
		$\{-1,2\}$ & $\{3,1\}$ & $3$ & quasi-conserved & $0*101$ & $\{3\}$ & \cite{daemen1995cipher,kriepke2026shift,liu2026finding} & $\{9\}$ & -- \\
		$\{-2,1\}$ & $\{3,-1\}$ & $3$ & quasi-conserved & $01*0-1$ & $\{3\}$ & -- & $\{9\}$ & -- \\
		$\{-2,4\}$ & $\{3,2\}$ & $3$ & quasi-conserved & $0-*-110$ & $\{3\}$ & -- & $\{9\}$ & -- \\
		$\{-2,1\}$ & $\{3,2\}$ & $3$ & quasi-conserved & $0-*011$ & $\{3\}$ & -- & $\{9\}$ & -- \\
		$\{-4,2\}$ & $\{3,-2\}$ & $3$ & quasi-conserved & $0-1-*-01$ & $\{3\}$ & -- & $\{9\}$ & -- \\
		$\{1,4\}$ & $\{-3,2\}$ & $-3$ & quasi-conserved & $1--*01-0$ & $\{3\}$ & -- & $\{9\}$ & -- \\
		$\{-1,2\}$ & $\{3,4\}$ & $3$ & quasi-conserved & $0*-011$ & $\{3\}$ & -- & $\{9\}$ & -- \\
		$\{-2,3\}$ & $\{4,1\}$ & $4$ & quasi-conserved & $0-*1-01$ & $\{4\}$ & -- & $\{12\}$ & -- \\
		$\{-2,2\}$ & $\{4,1\}$ & $4$ & quasi-conserved & $0-*10-1$ & $\{4\}$ & -- & $\{12\}$ & -- \\
		$\{-1,2\}$ & $\{4,1\}$ & $4$ & quasi-conserved & $0*10-1$ & $\{4\}$ & -- & $\{12\}$ & -- \\
		$\{-2,1\}$ & $\{4,-1\}$ & $4$ & quasi-conserved & $01*0--1$ & $\{4\}$ & -- & $\{12\}$ & -- \\
		$\{-2,2\}$ & $\{4,-1\}$ & $4$ & quasi-conserved & $01*-0-1$ & $\{4\}$ & -- & $\{12\}$ & -- \\
		$\{1,2\}$ & $\{4,3\}$ & $4$ & quasi-conserved & $*0011$ & $\{4\}$ & \cite{daemen1995cipher,kriepke2026shift,liu2026finding} & $\{12\}$ & -- \\
		\bottomrule
	\end{tabular}
\end{table}

\section{Conclusion and Future Work}\label{sec:con}

In this paper, we studied shift-invariant permutations arising from quasi-conserved landscapes on $\mathbb F_2^n$. For a general landscape function $\ell$ and the corresponding sequence of functions $\{\gamma_j\}_{j\geq 0}$ from $\mathbb F_2^n$ to itself, we proved that $\{\gamma_j\}_{j\geq 0}$ satisfies the polynomial composition property if and only if $\ell$ is quasi-conserved when $n>\max\{2m+k,2k+m\}$. Then, we studied the long-term behavior of the sequence of functions $\{\gamma_j\}_{j\geq 0}$.
Based on the polynomial composition property and the results on long-term behavior of $\{\gamma_j\}_{j\geq 0}$, we established an explicit isomorphism between the permutation group $\mathcal G^*$ constructed from $\{\gamma_j\}_{j\geq 0}$ and the unit group of $\mathbb F_2[z]/\langle p(z)\rangle$, where $p(z)$ is a monomial or binomial determined by the landscape function $\ell$. This isomorphism reduced the formulas for inverses, iterates and composition order of functions in the group to polynomial computations in the quotient ring. To apply these results to concrete examples, we specialized them to permutations represented by  binomials and trinomials in $\mathbb F_2[z]/\langle p(z)\rangle$ and determined $\xi(\beta_1)$ and $\xi(\kappa_1)$, which specify the values of $n$ for which $\beta_1$ and $\kappa_1$, respectively, are or are not permutations of $\mathbb F_2^n$. Furthermore, more than 100 representative constructions of shift-invariant permutations 
were listed in Tables~\ref{tab:landscapes}, \ref{tab:landscapes-Tq}, and~\ref{tab:landscapes-Tqt}. These families include several constructions previously studied in~\cite{daemen1995cipher,haugland2026new,kriepke2024algebraic,kriepke2026shift,liu2022inverse,liu2026finding,lyu2025generalized}.

By the definition of the landscape function $\ell$, it follows naturally that $n> m+k$. When $m+k<n\le\max\{2m+k,2k+m\}$, quasi-conservation of $\ell$ on $\mathbb F_2^n$ does not necessarily imply the polynomial composition property of $\{\gamma_j\}_{j\geq 0}$, as shown in Example~\ref{ex:finite-pcp-small-n}. In this case, we can define $\widetilde\ell$ on $\mathbb F_2^{\mathbb Z}$ using the formulas for $\ell$. All coordinate indices are taken in $\mathbb Z$, as described in the Appendix. By Remark~\ref{rem:integer-qc-bound}, if $\widetilde\ell$ is quasi-conserved on $\mathbb F_2^{\mathbb Z}$, then $\{\gamma_j\}_{j\ge0}$ satisfies the polynomial composition property on $\mathbb F_2^n$ when $n> m+k$. A direction for future work is to generalize the recurrence in \eqref{eq:gammaj}, or relax the definition of quasi-conservation to study permutation constructions that retain a weaker algebraic structure under composition. Developing permutation criteria and methods for computing inverses and iterates in these settings may also help explain further families listed in~\cite{daemen1995cipher}. Another direction for future work is to systematically analyze the cryptographic properties of the constructed shift-invariant permutations, including their differential properties and nonlinearity.

\clearpage

\appendix[Quasi-Conservation and Polynomial Composition on $\mathbb F_2^{\mathbb Z}$]\label{app:pcp}

The main text studies quasi-conservation and the polynomial composition property directly on $\mathbb F_2^n$, while Kriepke~\cite{kriepke2026shift} formulates these properties on $\mathbb F_2^{\mathbb Z}$.  In this appendix, we relate the finite-dimensional equivalence in Theorem~\ref{thm:equivalence} to Kriepke's infinite-dimensional characterization. We prove that the four statements in Theorem~\ref{thm:four-equivalence} are equivalent under an explicit dimension condition. We also give counterexamples showing that this dimension bound is sharp and that the assumption $q\ne0$ cannot be omitted.
	
Let $\mathbb F_2^{\mathbb Z}$ denote the set of functions $\mathbb Z\to\mathbb F_2$, written as two-sided sequences $x=(x_i)_{i\in\mathbb Z}$. The left shift is defined by $(Sx)_i=x_{i+1}$ for $i\in\mathbb Z$. A mapping $F:\mathbb F_2^{\mathbb Z}\to\mathbb F_2^{\mathbb Z}$ is shift-invariant if $F\circ S=S\circ F$. Unless stated otherwise, let $R,T\subseteq[-k,m]\setminus\{0\}$ be the same disjoint nonempty sets as in the main text and fix $q\in T$. Define $\widetilde\ell:\mathbb F_2^{\mathbb Z}\to\mathbb F_2^{\mathbb Z}$ by
\[\widetilde\ell(x)_i=\prod_{r\in R}(x_{i+r}+1)\prod_{t\in T}x_{i+t}\]
for every $i\in\mathbb Z$. Similarly, $\widetilde\ell^{(q)}$ and each $\widetilde\gamma_j$ $(j\ge0)$ map $\mathbb F_2^{\mathbb Z}$ to itself and are defined by the same formulas as in~\eqref{eq:ellq} and~\eqref{eq:gammaj}, respectively, with all coordinate indices taken in $\mathbb Z$. On $\mathbb F_2^{\mathbb Z}$, $\widetilde\ell$ is quasi-conserved with special index $q$ if and only if the following condition holds for every $\delta\in (R\cup T)\setminus\{q\}$:
\[(R+\delta)\cap T\ne\varnothing \qquad\text{or}\qquad (T+\delta)\cap R\ne\varnothing, \]
where all operations are performed in $\mathbb Z$. The polynomial composition property on $\mathbb F_2^{\mathbb Z}$ means that, for every $j,s\ge0$ and all
$\alpha_0,\ldots,\alpha_s\in\mathbb F_2$ with $\alpha_0=1$,
\begin{equation}\label{eq:polynomial-composition-property-Z}
		\widetilde\gamma_j\circ\left(
		\widetilde\gamma_0+\sum_{h=1}^{s}\alpha_h\widetilde\gamma_h
		\right)
		=\sum_{h=0}^{s}\alpha_h\widetilde\gamma_{h+j}.
\end{equation}
\begin{theorem}\label{thm:four-equivalence}
Let $n>\max\{2m+k,2k+m\}$. The following four statements are equivalent:
\begin{enumerate}
			\item[(i)] $\widetilde\ell$ is quasi-conserved with special index $q$ on $\mathbb F_2^{\mathbb Z}$;
			\item[(ii)] $\ell$ is quasi-conserved with special index $q$ on $\mathbb F_2^n$;
			\item[(iii)] $\{\widetilde\gamma_j\}_{j\ge0}$ satisfies the polynomial composition property on $\mathbb F_2^{\mathbb Z}$;
			\item[(iv)] $\{\gamma_j\}_{j\ge0}$ satisfies the polynomial composition property on $\mathbb F_2^n$.
\end{enumerate}
\end{theorem}
\begin{IEEEproof}
	Theorem~\ref{thm:equivalence} proves that (ii) and (iv) are equivalent. Reducing the set conditions in (i) modulo $n$, we obtain (ii) by Lemma~\ref{lem:quasi-cons-set}. For any $x\in\mathbb F_2^n$, repeat its $n$ coordinates in both directions to obtain an infinite $n$-periodic sequence in $\mathbb F_2^{\mathbb Z}$. For each $i\in\{0,\ldots,n-1\}$, the $i$-th coordinate of the image of this sequence under $\widetilde\gamma_j$ is $\gamma_j(x)_i$. Restricting the polynomial composition identities in (iii) to these $n$-periodic sequences shows that (iii) implies (iv). Thus it remains to prove that (ii) implies (i) and that (i) implies (iii).
	
Suppose that (ii) holds. For each $\delta\in (R\cup T)\setminus\{q\}$, Lemma~\ref{lem:quasi-cons-set} gives $r\in R$ and $t\in T$ such that $r+\delta\equiv t\pmod n$ or $t+\delta\equiv r\pmod n$.  Then $r+\delta=t$ or $t+\delta=r$ by
$|r+\delta-t|,|t+\delta-r|\le\max\{2m+k,2k+m\}<n$, so $(R+\delta)\cap T\ne\varnothing$ or $(T+\delta)\cap R\ne\varnothing$. This proves (i).
	
Suppose that (i) holds. Fix $j,s\ge0$ and $\alpha_0,\ldots,\alpha_s\in\mathbb F_2$ with $\alpha_0=1$. We show that~\eqref{eq:polynomial-composition-property-Z} holds on $\mathbb F_2^{\mathbb Z}$. Fix $x\in\mathbb F_2^{\mathbb Z}$ and $i\in\mathbb Z$. Both sides of this identity at coordinate $i$ use only finitely many input coordinates. Choose $N>\max\{2m+k,2k+m\}$ so that these input indices are distinct modulo $N$. Choose $y\in\mathbb F_2^N$ taking the same values as $x$ at these coordinates, where the coordinate subscripts of $y$ are taken modulo $N$. By the implication from (i) to (ii) (with $n$ replaced by $N$), $\ell$ is quasi-conserved on $\mathbb F_2^N$. Thus
\[ \widetilde\gamma_j\left(x+\sum_{h=1}^{s}\alpha_h\widetilde\gamma_h(x)\right)_i
		=\gamma_j\left(y+\sum_{h=1}^{s}\alpha_h\gamma_h(y)\right)_i
		=\sum_{h=0}^{s}\alpha_h\gamma_{h+j}(y)_i
		=\sum_{h=0}^{s}\alpha_h\widetilde\gamma_{h+j}(x)_i.
		\]
The first and third equalities follow from the choice of $y$, and the second follows from the established implication from (ii) to (iv). This proves (iii).
\end{IEEEproof}

Theorem~\ref{thm:four-equivalence} does not hold in general when $n=\max\{2m+k,2k+m\}$.

\begin{example}\label{ex:boundary-four-conditions}
	Take $n=3$, $R=\{-1\}$, $T=\{1\}$, and $q=1$. Here $k=m=1$ and $n=\max\{2m+k,2k+m\}$. On $\mathbb F_2^3$, we have
	\[
	(S^{-1}\ell\odot\ell)(x)_i
	=(1+x_{i-1})x_{i+1}(1+x_{i+1})x_i=0.
	\]
Since $(R\cup T)\setminus\{q\}=\{-1\}$, statement (ii) holds.  On $\mathbb F_2^{\mathbb Z}$, we have $(R-1)\cap T=\varnothing$ and $(T-1)\cap R=\varnothing$, so statement (i) fails.
\end{example}

\begin{remark}\label{rem:zero-special-index}
The assumption $q\ne0$ is needed for quasi-conservation to imply the polynomial composition property. To illustrate this, we temporarily allow $0\in T$. For example, take $R=\{1\}$, $T=\{0\}$, and $q=0$. Then $\widetilde\ell(x)_i=x_i(1+x_{i+1})$ is quasi-conserved and $\widetilde\gamma_j=\widetilde\ell$ for $j\ge1$. Choose $x\in\mathbb F_2^{\mathbb Z}$ with $(x_i,x_{i+1},x_{i+2})=(1,1,0)$. Then $\widetilde\ell(x)_i=0$ and $\widetilde\ell(x)_{i+1}=1$, so $\widetilde\ell(x+\widetilde\ell(x))_i=1$, whereas the polynomial composition property requires $\widetilde\ell\circ(\operatorname{id}+\widetilde\ell)=\widetilde\gamma_1+\widetilde\gamma_2=\mathbf0$.
\end{remark}

\begin{remark}\label{rem:integer-qc-bound} 
If $\widetilde\ell$ is quasi-conserved with special index $q\ne0$ on $\mathbb F_2^{\mathbb Z}$, then $\{\gamma_j\}_{j\ge0}$ satisfies the polynomial composition property on $\mathbb F_2^n$ for every $n>m+k$. In particular, if we assume quasi-conservation on $\mathbb F_2^{\mathbb Z}$ instead of $\mathbb F_2^n$, the bound $n>\max\{2m+k,2k+m\}$ in Section~\ref{sec:families} can be replaced by $n>m+k$, including the range $m+k<n\le\max\{2m+k,2k+m\}$.
		
Indeed, the implication from (i) to (iii) in Theorem~\ref{thm:four-equivalence} gives the polynomial composition property on $\mathbb F_2^{\mathbb Z}$. The equivalence between (i) and (iii) is also proved in Kriepke's doctoral thesis manuscript~\cite{kriepke2026thesis} by using a different method. Restricting these identities to the $n$-periodic sequences gives the polynomial composition identities on $\mathbb F_2^n$. The condition $n>m+k$ ensures that the offsets in 
$[-k,m]$ are distinct modulo $n$. 
\end{remark}

\end{document}